\documentclass[twoside]{article}

\usepackage[accepted]{aistats2023}
\usepackage[round]{natbib}

\usepackage[colorlinks=true,linkcolor=black,anchorcolor=black,citecolor=black,filecolor=black,menucolor=black,runcolor=black,urlcolor=black]{hyperref}

\usepackage{xcolor}

\usepackage{multirow}

\usepackage{microtype}
\usepackage{amsfonts}
\usepackage{mathtools}
\usepackage{cleveref}
\usepackage{amsthm}
\usepackage{amssymb}

\def\N{\mathbb{N}}
\def\R{\mathbb{R}}

\def\MC{\hat{\Pi}_{\textup{MC}}}
\def\MLMC{\hat{\Pi}_{\textup{MLMC}}}
\def\BQ{\hat{\Pi}_{\textup{BQ}}}

\def\MLBQ{\hat{\Pi}_{\textup{MLBQ}}}
\def\MSE{\textup{MSE}}
\def\Cost{\textup{Cost}}
\def\E{\mathbb{E}}
\def\V{\mathbb{V}}
\DeclareMathOperator*{\argmin}{argmin}
\def\Err{\textup{Err}}

\newenvironment{talign*}
{\csname align*\endcsname}
{\endalign}
\newenvironment{talign}
{\align}
{\endalign}

\newtheorem{theorem}{Theorem}

\newtheorem{proposition}[theorem]{Proposition}
\theoremstyle{definition}

\begin{document}

%

%
\runningauthor{Kaiyu Li, Daniel Giles, Toni Karvonen, Serge Guillas, Fran\c{c}ois-Xavier Briol}

\twocolumn[

\aistatstitle{Multilevel Bayesian Quadrature}

\aistatsauthor{ Kaiyu Li \And  Daniel Giles  \And Toni Karvonen}

\aistatsaddress{University College London \And University College London \And  University of Helsinki}

\aistatsauthor{Serge Guillas \And Fran\c{c}ois-Xavier Briol}
\aistatsaddress{ University College London \\ The Alan Turing Institute \And University College London \\ The Alan Turing Institute}
]

\begin{abstract}
Multilevel Monte Carlo is a key tool for approximating integrals involving expensive scientific models. The idea is to use  approximations of the integrand to construct an estimator with improved accuracy over classical Monte Carlo. We propose to further enhance multilevel Monte Carlo through Bayesian surrogate models of the integrand, focusing on Gaussian process models and the associated Bayesian quadrature estimators. We show, using both theory and numerical experiments, that our approach can lead to significant improvements in accuracy when the integrand is expensive and smooth, and when the dimensionality is small or moderate. We conclude the paper with a case study illustrating the potential impact of our method in landslide-generated tsunami modelling, where the cost of each integrand evaluation is typically too large for operational settings.  

\end{abstract}

\section{INTRODUCTION}
\label{sec:Intro}

This paper considers the task of approximating an unknown integral, or expectation, when evaluations of the integrand are expensive, either from a computational or financial point of view. This is a common problem in statistics and machine learning, where one commonly needs to marginalise random variables, compute normalisation constants of probability density functions or compute posterior expectations. However the problem is even more pronounced when doing uncertainty quantification for large mathematical models in science and engineering. For example, a scientist might be uncertain about the value of certain model parameters, and might therefore wish to estimate the expected value of some quantity of interest involving the model with respect to distributions on these parameters. 
\begin{figure}[t!]
\centering
\includegraphics[width=0.23\textwidth]{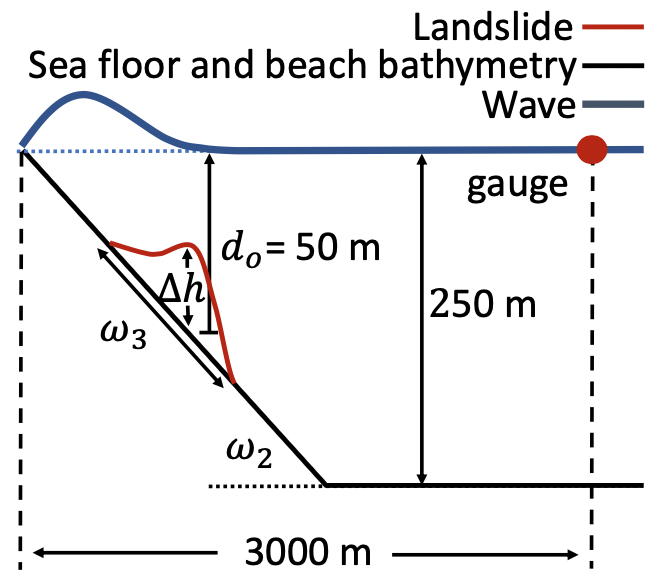}%
      {\includegraphics[width=0.53\linewidth]{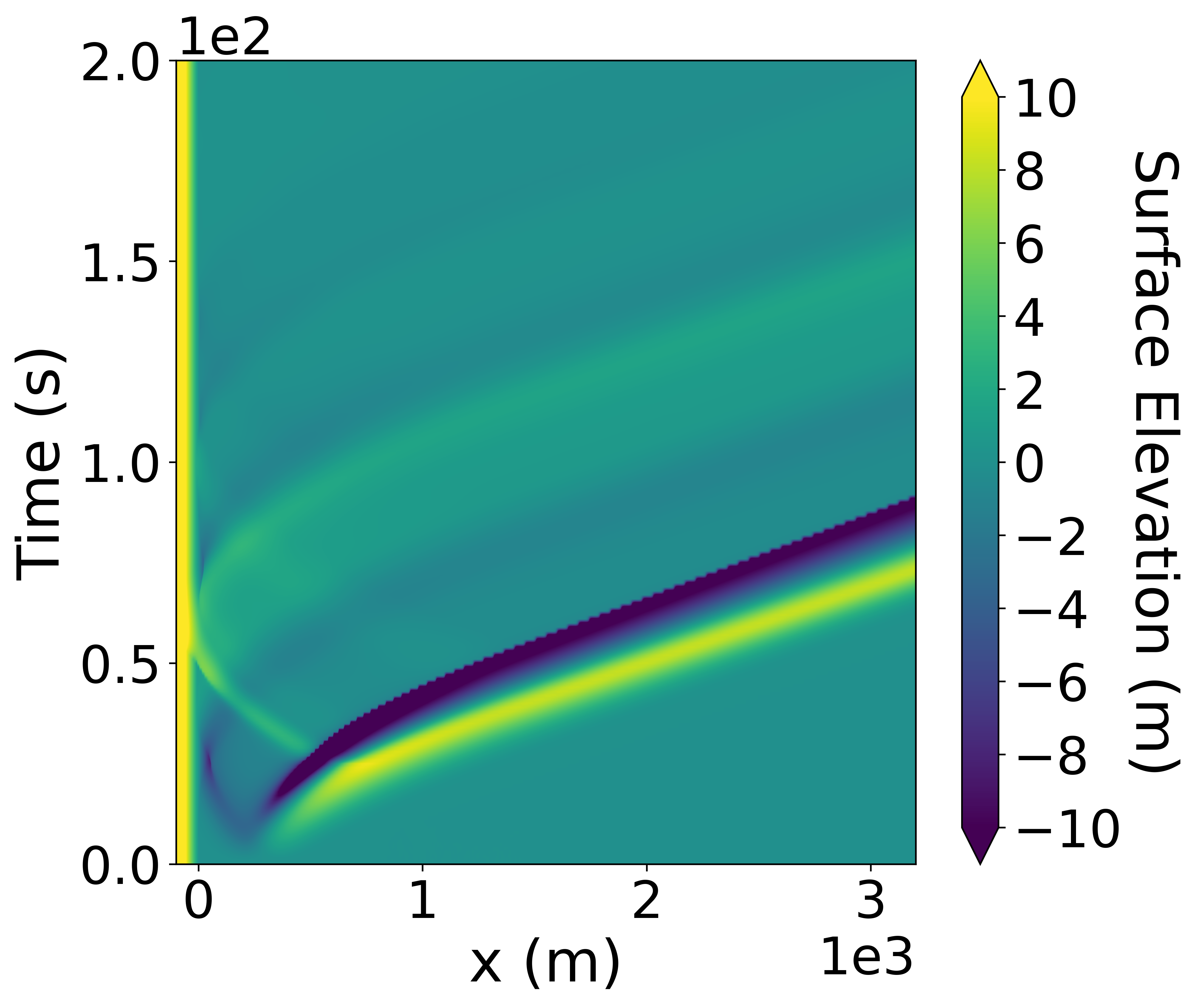}}%
    \label{fig:sketch}%
\caption{Tsunami Model. \emph{Left}: Sketch of the submerged landslide-generated tsunami. \emph{Right:} Solution of the differential equation through time and space.}
\vspace{-2mm}
\end{figure}

An example which illustrates this problem (later revisited in \Cref{sec:result}) is the modelling of landslide-generated tsunamis, where the evolution of the wave through space and time is described through a complex system of differential equations \citep{behrens2015tsunami, reguly2018volna, giles2020performance, Marras2021tsunami}; see Figure~\ref{fig:sketch} for an illustration. In this context, designers of tsunami resistant buildings, prevention structures or early warning systems might be interested in estimating the total wave energy or momentum flux of the tsunami at a fixed location. These quantities are functions of the solution of the differential equations, but there will usually be some uncertainty associated with certain physical parameters, such as those characterising the slope or size of the landslide. This uncertainty is represented through probability distributions, leading to the need to compute the expected value of the quantities of interest. The main challenge is that in order to obtain high accuracy estimates, it is necessary to use very fine time and space meshes to solve the differential equations, leading to prohibitively large computational costs.

A common approach to the approximation of such integrals is Monte Carlo (MC) methods, which include a wide range of simulation-based algorithms. Of particular relevance is \emph{multilevel Monte Carlo} (MLMC) \citep{giles2015multilevel} and its various extensions \citep{giles2009multilevel,Dick2011,kuo2015multi,kuo2017multilevel}. MLMC is designed for expensive integrands where cheap approximations are available at several levels of accuracy. Such models are called \emph{multifidelity models}  \citep{Peherstorfer2018}, and are widely used, including for atmospheric dispersion modelling \citep{katsiolides2018multilevel}, biochemical reaction network modelling \citep{Warne2019}, reliability theory \citep{Aslett2017}, erosion and flood risk modelling \citep{clare2021assessing}, pricing in finance \citep{dempster2018high}, wind farm modelling \citep{Kirby2023}, the design of advanced aerospace vehicles \citep{geraci2017multifidelity}, or tsunami modelling \citep{sanchez2016uncertainty}. 

MLMC evaluates the cheap but inaccurate approximation of the integrand a large number of times, and only evaluates the high-accuracy but expensive approximation of the integrand a small number of times. For the tsunami example above,  standard MC would use a fine time and space mesh, and evaluate the integrand at fixed high accuracy level. In contrast, MLMC will use several approximations with different meshes (each corresponding to a level), and use fewer evaluations of the expensive levels. For a fixed computational budget, this allows MLMC to obtain much more accurate estimate than standard MC. Beyond the scientific application areas above, this has also led MLMC to be used to enhance computational tools including Markov chain Monte Carlo \citep{Dodwell2019,Wang2022}, particle filters  \citep{gregory2017seamless}, approximate Bayesian computation \citep{Jasra2019}, Bayesian experimental design \citep{Goda2020} or variational inference \citep{Shi2021,Fujisawa2021}.

Unfortunately, most multilevel methods suffer from the fact that they are simulation-based methods which neglect all known properties of the integrand. This makes the methods widely applicable, but  means that their convergence rate will be slow when the integrand satisfies nice regularity conditions. This is clearly sub-optimal when working with  expensive models, where the number of evaluations will be limited. In this work, we propose to enhance MLMC through the use of surrogate models which encodes properties of the integrand, such as smoothness, sparsity or even periodicity. We focus in particular on Gaussian processes (GPs), which naturally lead to a class of algorithms that we call \emph{multilevel Bayesian quadrature} (MLBQ). 

MLBQ is a Bayesian probabilistic numerical method \citep{Hennig2015,Cockayne2017BPNM,Wenger2021,Hennig2022}, and more specifically a specific Bayesian quadrature algorithm (BQ; \citealp{Diaconis1988,o1991bayes,Rasmussen2003}); see \cite{briol2019probabilistic} for a recent overview. As we will see in the remainder of the paper, this approach can lead to a posterior distribution on the value of the integral, with (i) significant improvements in accuracy over existing methods when using the posterior mean as a point estimate, and (ii) the ability to quantify our uncertainty (given limited integrand evaluations) over the value of the integral.

\section{BACKGROUND}
\label{sec:background}

We now review key components of our approach: MC, multilevel models, MLMC and BQ.

\paragraph{Monte Carlo Methods}
Let $\Pi$ be a probability distribution on $\Omega \subseteq \mathbb{R}^d$ ($d \in \mathbb{N}_+$) and let $f \colon \Omega \to \mathbb{R}$ be some integrand of interest.
We focus on approximating
\begin{talign*}
    \Pi[f] \coloneqq \int_\Omega f(\omega) \Pi(d\omega)
\end{talign*}
and assume that $f$ is square integrable with respect to $\Pi$ (i.e.\ $\Pi[f^2] < \infty$). To tackle this task, we use pointwise evaluations of $f$: $\{\omega_i,f(\omega_i)\}_{i=1}^n$ for $n \in \N_+$ and $\omega_i \in \Omega$ for $i \in \{1,\ldots,n\}$. For example, a MC estimator \citep{robert2004monte,rubinstein2016simulation} takes the form
\begin{talign*}
     \MC[f] \coloneqq \frac{1}{n} \sum_{i=1}^n f(\omega_i), 
\end{talign*}
where $\{\omega_i\}_{i=1}^n \sim \Pi$; that is, $\{\omega_i\}_{i=1}^n$ are independent and identically distributed (IID) realisations from $\Pi$.
As $n \rightarrow \infty$ and under mild regularity conditions, MC estimators converge to $\Pi[f]$, making these approaches widely applicable. However, their performance when $n$ is finite and relatively small can be quite poor, which is a common issue when $f$ is expensive to evaluate. Alternative equal-weight estimators suffering from similar drawbacks include quasi-Monte Carlo (QMC) or randomised QMC \citep{owen2013monte}, which use $\{\omega_i\}_{i=1}^n$ that form a space-filling design.

\paragraph{Multilevel Monte Carlo}

For multifidelity models, we can improve on MC through MLMC.
Suppose that $f_L=f$, and $f_l \colon \Omega \rightarrow \R$ for $l \in \{0,\ldots,L-1\}$ are approximations of $f$ which increase both in accuracy and cost with the level $l$.
The integral of interest can be expressed through a telescoping sum as
\begin{talign}\label{eq:telescopic}
    \Pi[f] =\Pi[f_L] &= \Pi[f_0]+\sum^L_{l=1} \Pi[f_l-f_{l-1}].
\end{talign}
Instead of using a single MC estimator for $\Pi[f]$, we can estimate each term in the sum separately. Suppose that $\{\{\omega_{(l, i)}\}_{i=1}^{n_l}\}_{l=0}^L \sim \Pi$, the MLMC estimator is
\begin{talign*}
   \MLMC[f] \coloneqq{}& \MC[f_0] + \sum_{l=1}^L \MC[f_l -f_{l-1} ] \\
   ={}& \frac{1}{n_0}\sum^{n_0}_{i=1}f_0(\omega_{(0,i)}) \\
   &+  \sum^L_{l=1}\frac{1}{n_l}\sum^{n_l}_{i=1} (f_l(\omega_{(l,i)})-f_{l-1}(\omega_{(l,i)})).
\end{talign*}
For expensive integrands, there are two main advantages to this approach over MC. Firstly, each integrand (but the first) in the telescoping sum is of the form $f_l - f_{l-1}$, which will have low variance since we expect $f_l \approx f_{l-1}$ and hence $\V[f_l - f_{l-1}] \approx \V[0] = 0$. As a result, a small $n_l$ is sufficient to estimate such terms accurately through MC. Secondly, we have assumed that the functions are cheaper to evaluate for small $l$, so some of the initial terms in the sum can be estimated accurately through MC estimation with a large $n_l$. 

These remarks can be made precise by considering the computational cost necessary to obtain a given accuracy $\varepsilon$, or equivalently a given mean-squared error (MSE) $\varepsilon^2$. For an estimator $\hat{\Pi}[f]$, denote by $\Cost(\hat{\Pi},\varepsilon)$ this cost and by $\MSE(\hat{\Pi}) \coloneqq \mathbb{E}[(\hat{\Pi}[f] -\Pi[f])^2]=\mathbb{V}[\hat{\Pi}[f]] + (\mathbb{E}[\hat{\Pi}[f]]-\Pi[f])^2$ the MSE, where $\E$ and $\V$ denote the mean and variance with respect to all random variables in the estimator. For MC, $\E[\MC[f]]=\Pi[f]$ and $
\MSE(\MC) = \V[\MC[f]] = n^{-1} \mathbb{V}[f]$. To achieve a MSE of $\varepsilon^2$, $n$ should  be at least $\varepsilon^{-2}\mathbb{V}[f]$. If $C$ is the computational cost per sample, a MSE of $\varepsilon^2$ will lead to $\Cost(\MC,\varepsilon) = \varepsilon^{-2}\mathbb{V}[f] C$.

As we will now see, MLMC can provide significant improvements over MC. Let $C_0$ denote the cost of $f_0$, $C_l$ the cost of  $f_l-f_{l-1}$, $V_0 = \V[f_0]$ and $V_l = \V[f_l - f_{l-1}]$. The total cost of MLMC is $\sum_{l=0}^L n_l C_l$. The MSE and cost to achieve a MSE of $\varepsilon^2$ are hence
\begin{talign*}
\MSE(\MLMC) &= \V[\MLMC[f]] = \sum_{l=0}^L n_l^{-1}V_l, \\
\Cost(\MLMC,\varepsilon) &= \varepsilon^{-2}(\sum^L_{l=0}\sqrt{V_lC_l})^2.
\end{talign*}
To compare this cost with that of MC, we will consider two cases.
Firstly, if $V_lC_l$ increases rapidly with levels, we will have  $\Cost(\MLMC,\varepsilon) \approx \varepsilon^{-2} V_LC_L$. Secondly, if $V_lC_l$ decreases rapidly with levels, $\Cost(\MLMC,\varepsilon) \approx \varepsilon^{-2} V_0C_0$. In contrast, for standard MC, assuming the cost of evaluating $f_L$ is similar to that of evaluating $f_L -f_{L-1}$ and the variance of the estimate is $\V[f] = \V[f_L] \approx \mathbb{V}[f_0]$,  we have $\Cost(\MC,\varepsilon) \approx \varepsilon^{-2}V_0C_L$. Since $V_0 > V_L$ and $C_L > C_0$, we will therefore have $\Cost(\MC,\varepsilon) > \Cost(\MLMC,\varepsilon)$ regardless of the behaviour of $V_l C_l$.

This analysis of MLMC can be extended to find the optimal sample sizes per level given a fixed computational cost $T$ (see Appendix \ref{appendix_mlmc_optimN} or \citealp[Section~1.3]{giles2015multilevel} for a similar analysis with optimal sample sizes for a fixed MSE):
\begin{talign*}
n^{\text{MLMC}} &= \left(n^{\text{MLMC}}_0, \ldots, n^{\text{MLMC}}_L\right)\\ &\coloneqq \left(D  \sqrt{\frac{V_0}{C_0}}, \ldots, D  \sqrt{\frac{V_L}{C_L}}\right)
\end{talign*} 
where $D = T (\sum_{l^\prime=0}^L \sqrt{V_{l^\prime} C_{l^\prime}})^{-1}$. Despite the potential advantages of the approach above, there are also limitations which prevent the direct use of $n^{\text{MLMC}}_l$. Firstly, $V_l$ is usually unknown. It could be estimated from data, but unfortunately estimates of $V_l$ may be unreliable if the sample size at level $l$ is small. Secondly, $f_L$ is usually an approximation to $f$ (as opposed to $f=f_L$). Thirdly, as for our tsunami example, the number of levels can often be chosen by the user and it is often unclear how to decide which approximations $f_0, \ldots, f_L$ to include.

\paragraph{Bayesian Quadrature}

Clearly, the MLMC estimator can lead to significant gains, but we note that it focuses solely on sampling from $\Pi$ and does not utilise properties of~$f$. This is in contrast to BQ, an approach to integration which is based on a GP model of $f$. GPs are widely used as models for deterministic but computationally expensive functions, especially in computer experiments \citep{santner2003design, sacks1989design} and in spatial statistics \citep{stein1999interpolation}. We will denote a GP by $\mathcal{GP}(m, c)$ to emphasise the mean function $m: \Omega \to \mathbb{R}$ and the (symmetric and positive semi-definite) covariance function $c \colon \Omega \times \Omega \to \mathbb{R}$ (also called kernel), which uniquely identify the model. Given a $\mathcal{GP}(m, c)$ prior on $f$ and some observations $\{\omega_i,f(\omega_i)\}_{i=1}^n$ at pairwise distinct $\{\omega_i\}_{i=1}^n \subset \Omega$ for some $n \in \N_+$, the posterior on $f$ is also a GP with mean and covariance \citep{williams2006gaussian}
\begin{talign*}
    \tilde{m}(\omega)&=m(\omega)+c(\omega,W)c(W,W)^{-1}(f(W)-m(W)), \nonumber \\
    \tilde{c}(\omega,\omega')&=c(\omega,\omega')-c(\omega,W)c(W,W)^{-1}c(W,\omega')
\end{talign*}
 for all $\omega,\omega' \in \Omega$. Here, we use vectorised notation: $W=(\omega_1,\omega_2,\ldots,\omega_n)^\top$, $f(W)=(f(\omega_1),f(\omega_2),\ldots,f(\omega_n))^\top$, $c(\omega, W)=c(W,\omega)^\top=(c(\omega,\omega_1),\ldots,c(\omega,\omega_n))$ and $(c(W,W))_{i,j}=c(\omega_i,\omega_j)$ for all $i,j \in \{1,\ldots,n\}$
 .
 
 Prior knowledge on $f$, such as smoothness and periodicity, can be incorporated by specifying $m$ and~$c$.  For example, the squared exponential covariance function $c_\text{SE}(\omega,\omega^\prime)=\exp\left( -\|\omega-\omega'\|_2^2/\gamma^2\right)$ with length-scale $\gamma > 0$ implies a prior belief that $f$ has infinitely many derivatives. Alternatively, the Mat\'ern covariance function $
c_{\text{Mat\'ern}}(\omega,\omega^\prime)= 2^{1-v} \Gamma^{-1}(v)\left(\sqrt{2v}\|\omega-\omega^\prime\|_2 / \gamma\right)^v K_v\left(\sqrt{2v}\|\omega-\omega^\prime\|_2/ \gamma\right)$ with smoothness $v>0$ and length-scale $\gamma>0$, where $K_v$ is a modified Bessel function of the second kind, implies a belief that $f$ is $\lceil v \rceil-1$ times differentiable.

BQ \citep{Diaconis1988,o1991bayes,Rasmussen2003,briol2019probabilistic} is an estimator for $\Pi[f]$ motivated through Bayesian inference. The idea is to specify a prior on $f$, obtain the posterior on $f$ given evaluations of $f$, then consider the implied (pushforward) posterior on $\Pi[f]$. The most common approach uses a $\mathcal{GP}(m, c)$ prior; in that case, the posterior on $\Pi[f]$ is Gaussian with mean and variance
\begin{talign*}
    \E_{\text{BQ}}[\Pi[f]] &= \BQ[f] = \Pi[\tilde{m}] \\
    &\hspace{-0.8cm}= \Pi[m]+\Pi[c(\cdot,W)]c(W,W)^{-1} (f(W)-m(W)),\nonumber \\
    \V_{\text{BQ}}[\Pi[f]]  &= \Pi [\Pi[\tilde{c}]] \\
    & =\Pi[\Pi[c]]-\Pi[c(\cdot,W)]c(W,W)^{-1}\Pi[c(W,\cdot)],  \nonumber
\end{talign*}
where $\Pi[c(\cdot,W)]=(\Pi[c(\cdot,\omega_1)],\ldots,\Pi[c(\cdot,\omega_n)])^\top$ and we use the convention that for a function with two inputs, $\Pi[\Pi[\cdot]]$ always denotes integration once with respect to each input. In contrast with MC methods which rely on central limit theorems, $\V_{\text{BQ}}[\Pi[f]]$ can quantify our uncertainty about $\Pi[f]$ for finite (and possibly small) $n$. 

A particular advantage of the formulae above is that they are defined for arbitrary $\{\omega_i\}_{i=1}^n$. A number of point sets have been studied including IID \citep{Rasmussen2003}, QMC \citep{briol2019probabilistic,Jagadeeswaran2018}, realisations from determinental point processes \citep{Bardenet2019}, point sets with symmetry properties \citep{Karvonen2017symmetric,Karvonen2019} and adaptive designs \citep{Osborne2012active,Gunter2014,briol2015frank}.
For specific point sets and GP priors, $\BQ[f]$ also coincides with classical quadrature rules \citep{Diaconis1988,Karvonen2017classical}.

The two main disadvantages of BQ are that: (i) as per GPs, the computational cost is $\mathcal{O}(n^3)$, due to the need to invert $n \times n$ matrices, and (ii) $\Pi[c(\cdot,\omega)]$ for $\omega \in \Omega$ and $\Pi[\Pi[c]]$ are only tractable for some pairs of distributions and covariance functions (see Table 1 in  \citealp{briol2019probabilistic}). On the other hand, BQ also has much faster convergence rates than classical Monte Carlo methods when $d$ is small or moderate \citep{briol2019probabilistic,kanagawa2020convergence,wynne2021convergence}. For this reason, BQ has mostly been applied to problems where $n$ is constrained to be small (for example when the integrand is expensive) and the integration measure is relatively simple. This includes problems in global illumination in computer graphics \citep{Brouillat2009}, cardiac modelling \citep{Oates2017heart}, engineering control \citep{Paul2016}, econometrics \citep{Oettershagen2017}, risk \citep{CADINI201615}, likelihood free inference \citep{Bharti2023} and in variational inference \citep{Acerbi2018}.

\section{METHODOLOGY}
\label{sec:method}

Although MLMC is particularly well-suited to integrals involving multifidelity models, it usually disregards any prior information on the integrand. 
We now remedy this issue by designing a novel estimator which combines the advantages of BQ and MLMC. Our proposed algorithm is relatively straightforward: it uses the telescopic sum in Equation \eqref{eq:telescopic} and approximates each of the terms through BQ rather than MC. Here and throughout the remainder of the paper, we use the convention that $f_{-1} \equiv 0$ to simplify all expressions. Suppose we have access to the evaluations $\{\{f_l(\omega_{(l,i)})-f_{l-1}(\omega_{(l,i)})\}_{i=0}^{n_l}\}_{l=0}^L$ of the approximate integrands on  $\Omega$. We will specify a sequence of priors such that $\mathcal{GP}(m_l,c_l)$ is a prior on the increment $f_l-f_{l-1}$, and we will take these increments to be independent a-priori.
\begin{proposition} 
\label{BMLMC_sum}
Given the priors and datasets described above, the posterior on $f$ is a Gaussian process and the posterior on $\Pi[f]$ is a univariate Gaussian with mean
\begin{talign*}
   \E_{\textup{MLBQ}}[\Pi[f]] &\coloneqq \sum_{l=0}^L \BQ[f_l -f_{l-1} ]  \\
     & \hspace{-1.8cm} = \sum_{l=0}^L \big( \Pi[m_l] + \Pi[c_l(\cdot,W_l)]c_l(W_l,W_l)^{-1} \big. \\ 
     &\hspace{1.15cm} \big. \times (f_l(W_l)-f_{l-1}(W_l)-m_l(W_l)) \big)
\end{talign*}
and variance
\begin{talign*}
    \V_{\textup{MLBQ}}[\Pi[f]] &\coloneqq \sum^{L}_{l=0}\mathbb{V}_{\textup{BQ}}[\Pi[f_l-f_{l-1}]]\\
   & \hspace{-1.9cm} = \! \sum_{l=0}^L \! \big( \Pi[\Pi[c_l]] \! - \! \Pi[c_l(\cdot,W_l)]c_l(W_l,W_l)^{-1}\Pi[c_l(W_l,\cdot)] \big)
\end{talign*}
where $W_l=(\omega_{(l,1)},\ldots,\omega_{(l,n_l)})^\top$ for $l \in \{0,\ldots,L\}$.
\end{proposition} 
The proof is given in Appendix \ref{append:ProofProp1}. Once again, a point estimator can be obtained through the posterior mean $\MLBQ[f] \coloneqq \E_{\text{MLBQ}}[\Pi[f]]$ and we will call this the \emph{multilevel Bayesian quadrature} (MLBQ) \emph{estimator}. Although MLBQ requires only a straightforward modification of the MLMC algorithm, we will see in the remainder of the paper that it will allow us to take advantage of the properties of both MLMC and BQ. 

A simple illustration example comparing BQ and MLBQ ($L=2$) with the same evaluation constraint is shown in Figure \ref{fig:Illustration}. We used the approximations from the Poisson equation experiment in Section \ref{sec:result} and Appendix \ref{Append: FEM}. As we observed, the GP for MLBQ fits $f_2$ better than the GP for BQ. The MLBQ estimator has smaller error and smaller variance than the BQ estimator.
\begin{figure*}[t!]
    \centering
    \includegraphics[width=\textwidth]{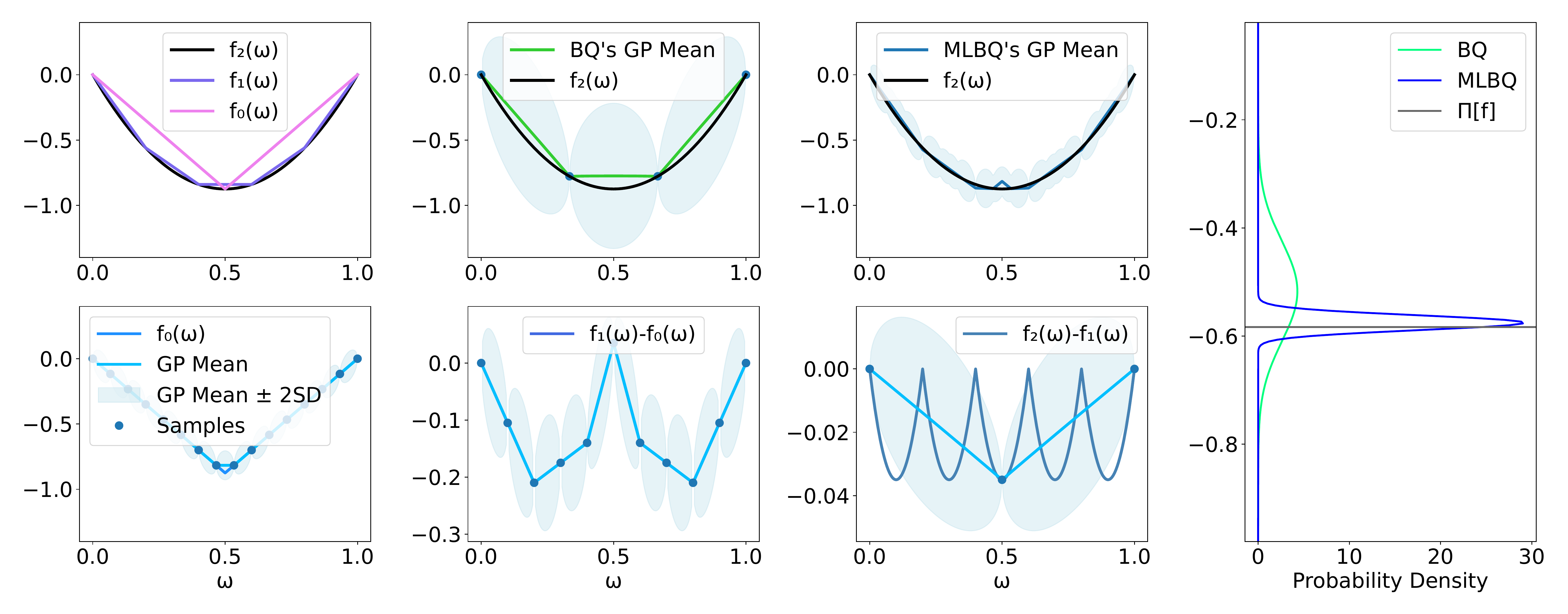}
    \caption{Illustration Example: \textit{Upper Left}: the approximations to $f$, GP for BQ, GP for MLBQ. \textit{Bottom Left}: GP for level 0, 1, 2 of MLBQ. \textit{Right}: BQ and MLBQ esitmators.}    \label{fig:Illustration}
\end{figure*}
The cost for implementing MLBQ is $\mathcal{O}(\sum_{l=0}^L n_l^3)$, which is larger than the $\mathcal{O}(\sum_{l=0}^L n_l)$ of MLMC. However, for most multifidelity models, we expect these costs to be dwarfed by the cost of function evaluations which is $\sum_{l=0}^L n_l C_l$. Additionally, we will see in the next Section that MLBQ can have a much faster convergence rate than MLMC. Due to the independence assumption, we can estimate the GP hyperparameters separately for each level; see Section \ref{append: hyper_prior}.  If the assumption is violated, we could be under- or over-estimating our uncertainty. It is possible to do away with this assumption by modelling levels jointly as demonstrated in Appendix \ref{append:correlated_case} or following the work on multi-output BQ of \citet{Xi2018MultiOutput}, but this would prohibitively increase the cost to $\mathcal{O}((\sum_{l=0}^L n_l)^3)$.

\section{THEORY}
\label{sec:theory}

We now prove an upper bound on the error of MLBQ and derive the optimal number of samples per level. 

Let $L^2(\Omega)$ denote the space of square-integrable functions on $\Omega \subseteq \mathbb{R}^d$ with respect to the Lebesgue measure.
The Sobolev space $W^\alpha_2(\Omega)$ of integer order $\alpha \geq 0$ consists of functions $f \in L^2(\Omega)$ for which $\|f\|_{\alpha} \coloneqq  (\sum_{\beta\in \mathbb{N}^d\, : \, \lvert\beta\rvert \leq \alpha } \|D^\beta f\|^2_{L^2(\Omega)} )^{1/2} < \infty$,
where $\lvert \beta \rvert = \sum^d_{i=1} \beta_i$ and $D^\beta f$ is the weak derivative~\citep[p.\@~22]{adams2003sobolev} of order $\beta$.
For non-integer $\alpha \geq 0$, the Sobolev norm can be defined via Fourier transforms and the two definitions coincide, up to a constant, for integer $\alpha$ if $\Omega$ is sufficiently regular~\citep[Section~2.2]{wynne2021convergence}.
The space $W_2^\alpha(\Omega)$ is a Hilbert space.

By the Moore--Aronszajn Theorem~\citep[Theorem~3 in Chapter~1]{Berlinet2004}, every positive semi-definite covariance function $c \colon \Omega \times \Omega \to \mathbb{R}$ induces a unique reproducing kernel Hilbert space (RKHS) $\mathcal{H}(c)$  consisting of functions $f \colon \Omega \to \mathbb{R}$ and equipped with an inner product $\langle \cdot, \cdot \rangle_{\mathcal{H}(c)}$ and norm $\| \cdot \|_{\mathcal{H}(c)}$.
 The RKHS satisfies: (1) $c(\cdot, \omega) \in \mathcal{H}(c)$ for every $\omega \in \Omega$, and (2) the reproducing property that $f(\omega)=\langle f,c(\cdot, \omega )\rangle_{\mathcal{H}(c)}$ for every $f \in \mathcal{H}(c)$ and $\omega \in \Omega$.
 
The following assumptions are used in our results:
\begin{enumerate}
\item[A1.] The domain is of the form $\Omega = \Omega_1 \times \cdots \times \Omega_d$ for each $\Omega_i$ a non-empty interval.
\item[A2.] The distribution $\Pi$ has a bounded density function $\pi$; i.e. $\| \pi \|_{L^\infty(\Omega)} \coloneqq \sup_{ \omega \in \Omega } \pi(\omega) < \infty$.
\item[A3.] For each $l \in \{0, \ldots, L\}$, the RKHS $\mathcal{H}_l \coloneqq \mathcal{H}(c_l)$ is norm-equivalent to  $W_2^{\alpha_l}(\Omega)$ for $\alpha_l > d/2$. Two Hilbert spaces $H_1$ and $H_2$ are norm-equivalent if and only if they are equal as sets and there are constants $b_1,b_2 >0$ such that $b_1 \| f \|_{H_1} \leq \| f \|_{H_2} \leq b_2 \| f \|_{H_1}$ for all $f \in H_1 = H_2$.
\item[A4.] There are $\beta_0,\ldots,\beta_L > d/2$ such that $f_0 \in W_2^{\beta_0}(\Omega)$ and $f_l, f_{l-1} \in W_2^{\beta_l}(\Omega)$ for every $l \in \{1, \ldots, L\}$.
\item[A5.] For each $l \in \{0, \ldots, L\}$, the fill-distance $h_{W_l, \Omega} = h_{l, \Omega} \coloneqq \sup_{ \omega \in \Omega} \min_{i = 1,\ldots, n_l} \| \omega - \omega_{(l,i)} \|_2$ satisfies $h_{l, \Omega} \leq h_\textup{qu} n_l^{-1/d}$ for a constant $h_\textup{qu} > 0$.
\item[A6.] The prior means are $m_l \equiv 0$ for all $l \in \{0, \ldots, L\}$.
\end{enumerate}

The purpose of Assumption~A1 is to ensure that the domain is sufficiently regular for the use of Sobolev extension and embedding theorems. This assumption could be generalised to allow more complex domains without affecting the convergence rate \citep[Section~3.1]{wynne2021convergence}.  Assumption~A3 and its relatives are standard in the error analysis of GP and BQ methods~\citep[e.g.,][]{karvonen2020maximum, teckentrup2020convergence, wynne2021convergence} and are important for deriving our theoretical results.
The RKHS of a Mat\'ern kernel $c_\textup{Mat\'ern}$ with smoothness $v$ and any length-scale is norm-equivalent to $W_2^{\alpha}(\Omega)$ for $\alpha = v + d/2$ whenever $\Omega$ satisfies Assumption~A1. Assumption~A5, known as the quasi-uniformity assumption~\citep[Section~14.1]{wendland2004scattered}, ensures that each of the sets $W_{l}$ covers $\Omega$ in a sufficiently uniform manner, because the fill-distance of a set $W$ equals the radius of the largest ball in $\Omega$ which contains no point from $W$. Regular grids are examples of sets that satisfy Assumption~A5. Assumption~A6 is made out of convenience and could be replaced with the assumption that $m_l \in W_2^{\beta_l}(\Omega)$ for each $l$.

\begin{theorem} \label{theo1}
Suppose that assumptions A1--A6 hold and define $\tau_l \coloneqq \min\{ \alpha_l, \beta_l\}$.
Then 
\begin{talign*} 
    \Err(\MLBQ) &= \lvert \Pi[f_L]-\MLBQ[f_L] \rvert \\
    &\leq \| \pi \|_{L^\infty(\Omega)} \sum^L_{l=0} a_l \|f_l-f_{l-1}\|_{\tau_l}n_l^{-\tau_l/d}\nonumber
\end{talign*}
whenever each $n_l$ is sufficiently large.
Each constant $a_l > 0$ depends on $\alpha_l$, $\beta_l$, $c_l$, $h_\textup{qu}$, $d$, and $\Omega$, but not on $f_l$ or the data points.
\end{theorem}

Theorem~\ref{theo1} is proved in Appendix \ref{append:ProofTheo1}. The proof is similar to the convergence proofs in~\citet{kanagawa2020convergence,karvonen2020maximum, teckentrup2020convergence, wynne2021convergence}.
The Sobolev norm in the bound may be replaced with the RKHS norm $\| f_l - f_{l-1} \|_{\mathcal{H}_l}$ if $\beta_l \geq \alpha_l$ due to assumption~A3.

If it is assumed that $\beta_l \geq \alpha_l$ for each $l$, one may use Theorem~1 and Corollary~2 in \citet{Krieg2022} to prove a variant of Theorem~\ref{theo1} in which the points are independent samples from a  uniform distribution on $\Omega$ and the upper bound is for the expected error of the MLBQ.

Various other generalisations of Theorem~\ref{theo1} are possible but are not included here so as to simplify the presentation of our assumptions. These include non-zero prior means, varying kernel parameters~\citep{teckentrup2020convergence}, misspecified likelihoods~\citep{wynne2021convergence}, and improved rates when each $f_l$ has, essentially, twice the smoothness of $c_l$ (\citealp{tuo2020improved};~\citealp[Sections~3.4 and~4.5]{karvonen2020maximum}) or when both $f_l$ and $c_l$ are infinitely differentiable \citep[Theorem~2.20]{karvonen2019kernel}.

At each level $l$, the convergence rate of $\mathcal{O}(n_l^{-\tau_l/d})$ is faster than the  rate for MC estimators of $\mathcal{O}(n_l^{-1/2})$ because $\tau_l/d = \min\{\alpha_l, \beta_l\} / d > 1/2$.
Since $f_0,...,f_L$ approximate the same function, the kernels $c_l$ and the smoothnesses $\alpha_l$ and $\beta_l$ do not typically change with $l$, which means that the the constants $a_l$ do not change.
If, additionally, $\| f_l - f_{l-1} \|_{\tau_l}$ tends to zero as $l$ increases, which is usually the case because approximation quality should increase with the level, we see that fewer evaluations are needed at higher levels.
However, if $\tau_l$ differ significantly, more evaluations than expected may be needed at higher levels.

Using Theorem \ref{theo1} and assuming we use the same prior at each level, we can also derive the optimal number of samples for MLBQ under a limited computational budget.
To do so, we assume that the cost of fitting GPs at each level is dwarfed by the cost of function evaluations. This is reasonable because function evaluation costs tend to be relatively large for applications where MLMC is commonly used. For example, for differential equation models the cost is usually driven by the cost of the solvers such as finite difference, finite element or finite volume methods, and this can be large for fine meshes. For example, for the tsunami example in \Cref{sec:Tsunami}, fitting all the GPs takes less than $25$ seconds whereas a single evaluation of $f_L-f_{L-1}$ takes $150$ seconds. For this reason, we therefore assume that the total cost of running MLBQ and functions evaluations is given by $\gamma \sum^L_{l=0} C_{l} n_{l}$ for some $\gamma \geq 1$ but close to $1$.

\begin{theorem} \label{theo2}
  Suppose that assumptions A1--A6 hold and $c_l$ and $\tau \coloneqq \tau_l = \min\{\alpha_l, \beta_l\} = \alpha_l$ do not depend on $l$. Then 
\begin{talign*}
n^{\textup{MLBQ}} &= \left(n_0^{\textup{MLBQ}},\ldots,n_L^{\textup{MLBQ}}\right) \\
&\coloneqq  \underset{\substack{n_0,n_1,\cdots,n_L \\ \text{s.t. } \gamma \sum^L_{l=0} C_{l} n_{l} = T }}\argmin   \sum^L_{l=0} a_l \|f_l-f_{l-1}\|_{\tau} n_l^{-\tau/d}   
\end{talign*}
  for $\gamma \geq 1$ and $T > 0$ is solved by
  \begin{talign*} 
    n^{\textup{MLBQ}}_l &= D \left( \frac{ \|f_l-f_{l-1}\|_{\tau}}{C_l} \right)^{\frac{d}{\tau + d}} \quad \forall l \in \{0, \ldots, L\},
  \end{talign*}
  where
  $D= T  \left(\gamma \sum^L_{{l'}=0}C_{l'}^{\frac{\tau}{\tau + d}} \left( \|f_{l'}-f_{{l'}-1}\|_{\tau} \right)^{\frac{d}{\tau + d}}\right)^{-1}$.
\end{theorem}

The proof is given in Appendix \ref{append:ProofTheo2}. The additional assumptions were introduced to simplify the result by ensuring that $a_l$ does not depend on $l$. If the function evaluation costs do not dominate or if $\tau_l$ differ, one can still calculate the optimal sample sizes by solving the optimisation problem in \Cref{theo2} numerically. The theorem provides a solution for the relaxed optimisation problem where $n_1,\ldots,n_L$ are real numbers. In practice, it will needed to use natural numbers, and this is possible by taking the floor or ceiling of each $n_1,\ldots,n_L$.

The optimal sample sizes for MLMC and MLBQ are similar; here, $\|f_l -f_{l-1}\|_\tau$ is analogous to $V_l$ in that it measures the size of each element in the telescoping sum. We expect $\|f_l-f_{l-1}\|_{\tau}$ to be a decreasing function of $l$ which converges to zero.  If the convergence is slow, the sample size for large $l$ has to be relatively large, whereas it can be relatively small otherwise. Additionally, a large cost $C_l$ also leads to relatively smaller sample sizes. For MLMC, the optimal sample size at level $l$ is proportional to $C_l^{-1/2}$ whereas for MLBQ it is proportional to $C_l^{-d/(\tau+d)}$. Therefore, when $\tau > d$, the penalisation for large $C_l$ is smaller for MLBQ than MLMC, and vice-versa. This is intuitive because when $\tau$ is large, the integrands are smoother and we expect BQ to be able to approximate them fast in the number of samples.

Plugging in the optimal samples sizes of \Cref{theo2} to the bound in \Cref{theo1}, we obtain that 
\begin{talign*}
\Err(\MLBQ) &\leq  A T^{-\frac{\tau}{d}} \left( \sum^L_{{l}=0} C_{l}^{\frac{\tau}{\tau + d}}  \|f_{l}-f_{{l}-1}\|_{\tau}^{\frac{d}{\tau + d}} \right)^{\frac{\tau+d}{d}},
\end{talign*}
where $A = \| \pi \|_{L^\infty(\Omega)} a \gamma^{\tau/d}$. 
For BQ based on evaluations of $f_L$ and utilising the same computational budget we obtain 
\begin{talign*}
\Err(\BQ) & \leq A T^{-\frac{\tau}{d}} C_L^{\frac{\tau}{d}}\|f_L\|_{\tau}
\end{talign*}
from \Cref{theo1} by setting $f_l \equiv 0$ and $C_l = 0$ for every $l \in \{0, \ldots, L - 1\}$.
Let us denote the two upper bounds above by $B_\textup{MLBQ}$ and $B_\textup{BQ}$.
To compare these bounds, we consider two cases. Firstly, if the term $b_l \coloneqq {C_l}^{\tau/(\tau+d)}\|f_l-f_{l-1}\|_{\tau}^{d/(\tau+d)}$ grows rapidly with $l$, then $B_\textup{MLBQ}$ is dominated by the highest level $L$, so that $B_\textup{MLBQ} \approx A T^{-\tau/d} C_L^{\tau/d} \|f_L-f_{L-1}\|_{\tau}$. Secondly, if $b_l$ decreases rapidly with $l$, then $B_\textup{MLBQ} \approx A T^{-\tau/d} C_0^{\tau/d} \|f_0\|_{\tau} $. In either case, the bound on $\Err(\MLBQ)$ is smaller than that on $\Err(\BQ)$ under natural assumptions. 
In the first case
\begin{talign*}
B_\textup{BQ} \approx \frac{\|f_L\|_{\tau}} { \|f_L-f_{L-1}\|_{\tau}} B_\textup{MLBQ} \geq B_\textup{MLBQ}
\end{talign*}
if $\|f_L\|_{\tau} \geq \|f_L-f_{L-1}\|_{\tau}$, whilst in the second case
\begin{talign*}
B_\textup{BQ} \approx \left(\frac{C_L}{C_0}\right)^{\tau/d} \frac{ \| f_L \|_\tau }{ \| f_0 \|_\tau } B_\textup{MLBQ} \geq B_\textup{MLBQ}
\end{talign*}
if $C_L \geq C_0$ and $\| f_L \|_\tau \geq \| f_0 \|_\tau$.

\section{PRACTICAL CONSIDERATIONS}\label{append: hyper_prior}
Before moving on to experimental results, we briefly discuss practical considerations for the implementation of MLBQ.

Firstly, for each level $l$, we will usually include at least one amplitude $\sigma_l$ parameter (so that the covariance function takes the form $\tilde{c}_l(\omega,\omega') = \sigma_l^2 c_l(\omega,\omega')$ for some covariance function $c_l$) and a lengthscale $\gamma_l$ (or a lengthscale per dimension of the data). We suggest to select these by maximising the marginal log-likelihood separately for each level:
\begin{talign*} 
   L(\gamma_l,\sigma_l) & = -\frac{1}{2} {\sigma_l}^{-2} \left(f_l(W_l)-f_{l-1}(W_l)\right)^\top c_l(W_l,W_l)^{-1}\\
   &\qquad \times \left(f_l(W_l)-f_{l-1}(W_l)\right)  -n_l \log {\sigma_l} ^2\\
   & \quad \qquad+\frac{1}{2}\log|c_l(W_l,W_l)|-\frac{n_l}{2}\log 2 \pi ,
\end{talign*}
For a given $\sigma_l$, this can be done in closed-form as follows:
\begin{talign*} 
     \sigma^*_l=\sqrt{\frac{(f_l(W_l)-f_{l-1}(W_l))^\top c_l(W_l,W_l)^{-1}(f_l(W_l)-f_{l-1}(W_l))}{n_l}}.
\end{talign*}
For the lengthscale, the maximum of $L(\gamma_l,\sigma_l)$ as a function of $\gamma_l$ needs to be obtained numerically. When $n_l$ is large, we can use mini-batches with stochastic optimization. 

Note that it is essential to select the hyperparameters for each level independently. To illustrate this, consider each level having prior $\mathcal{GP}(0, \sigma^2 c_l)$. All other parameters besides the amplitude are fixed, and maximum likelihood is used to estimate the amplitude.
The resulting maximum marginal likelihood estimate (MLE) is
\begin{talign*}
  \sigma_\textup{all} = \sqrt{ \frac{(y^\top c(W, W)^{-1} y}{\sum_{l=0}^L n_l} },
\end{talign*}
where the vectors $W$ and $y$ are formed by stacking all $W_l$ and $f_l(W_l) - f_{l-1}(W_l)$, respectively and the matrix $c(W,W)$ is formed with diagonal blocks all $c_l(W_l, W_l)$ and off-diagonal components $0$s.
Inserting this MLE in the equation for the MLBQ variance yields
\begin{talign} \label{variance-sigma-all}
  \V_{\textup{MLBQ}}[&\Pi[f]] = \frac{y^\top c(W, W)^{-1} y}{\sum_{l=0}^L n_l} \sum_{l=0}^L \big(  \Pi[\Pi[c_l(\cdot,\cdot)]]  \nonumber\\
  &-\Pi[c_l(\cdot,W_l)]c_l(W_l,W_l)^{-1}\Pi[c_l(W_l,\cdot)] \big).
\end{talign}
Because each term in the sum above depends only on $W_l$, the knowledge that $f_l - f_{l-1}$ tends to zero as $l$ increases is not exploited. The essential property of a multilevel method that less data is needed on higher levels is not reflected in the MLBQ variance.
This defect is eliminated if each level has an independently estimated amplitude parameter.
In this case the variance becomes
\begin{talign*}
  \V_{\textup{MLBQ}}[\Pi[f]] &= \sum_{l=0}^L \sigma_l^*{}^2 \big( \Pi[\Pi[c_l(\cdot,\cdot)]]\\
  & \quad-\Pi[c_l(\cdot,W_l)]c_l(W_l,W_l)^{-1}\Pi[c_l(W_l,\cdot)] \big).
\end{talign*}
Now the magnitude of $f_l - f_{l-1}$ directly affects the $l$th term: if $f_l - f_{l-1}$ is small, the contribution of the $l$th term to $\V_{\textup{MLBQ}}[\Pi[f]]$ is small even if $W_l$ contains only few points, unlike in \eqref{variance-sigma-all}.

\section{EXPERIMENTS}
\label{sec:result}
We now evaluate MLBQ for synthetic differential equation models and landslide-generated tsunami modelling. The code to reproduce our experiments is available at \url{https://github.com/CeciliaKaiyu/MLBQ}. The MLBQ method is also implemented in the \texttt{ProbNum} open-source Python package \citep{Wenger2021}.

\paragraph{Poisson Equation} \label{sec:PE example}

The Poisson equation is a canonical partial differential equation which arises in physics~\citep[e.g.,][Chapter~8]{MathewsWalker1970}. We consider a synthetic model where for $f \colon (0, 1) \to \R$, 
\begin{talign*}
  f''(\omega) &= z(\omega) \: \text{ for } \:  \omega \in (0,1) \quad \& \quad
    f(0)=f(1)=0
\end{talign*}
where $z(\omega)=1$. Here, $\Pi[f]=\int_0^1 f(\omega) d\omega$ so that $\Pi$ is a $\text{Unif}(0,1)$. To obtain $f_0,\ldots,f_L$, we use piecewise linear finite element approximations as described in Appendix \ref{Append: FEM}. We use $L=2$ and have $C=(C_0,C_1,C_2) = (3.6,8.5, 42.4)$ (all measured in $10^{-3}$ seconds). This problem is relatively simple and could be brute-forced with MC, but has the advantage that we can compute the optimal sample sizes for MLBQ and MLMC since A1--A6 are all satisfied when using a unifom grid of points and $\|f_l-f_{l-1}\|_{\tau_l}$ can be computed in closed form for all $l$. It therefore makes for a good test-bed for our method.

\begin{figure}[h]
    \centering
 \includegraphics[width=0.49\textwidth]{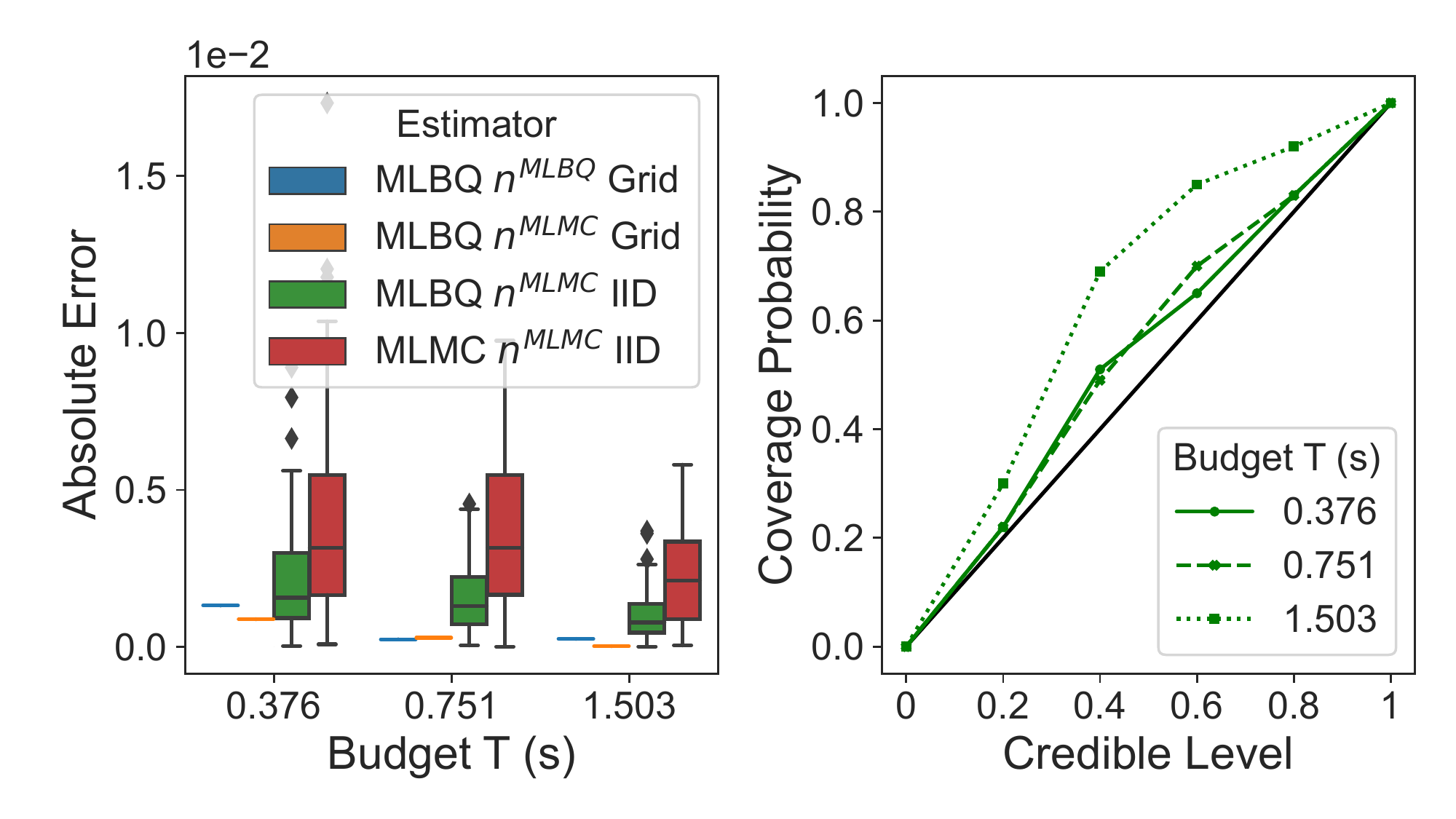}
    \caption{Poisson Equation. \textit{Left}: Absolute integration error. \textit{Right}: Calibration of MLBQ, IID points.}
    \label{fig:PE3L}
\end{figure}

We compare four different settings: MLBQ using $n^{\text{MLBQ}}$ and uniform grid points, MLBQ using $n^{\text{MLMC}}$ and uniform grid points  or IID points, and MLMC using $n^{\text{MLMC}}$ and IID points. To implement $n^{\text{MLMC}}$, we brute-forced the computation of $V_0,\ldots,V_L$ through an MC approximation. All MLBQ algorithms use a mean-zero GP with Mat\'ern $0.5$ kernel, and all sample sizes are given in Appendix \ref{append: Experimental_Details_PE}. 

Figure \ref{fig:PE3L} visualizes the result of $100$ repetitions of the experiment, where for each repetition, we evaluated $f_0,\ldots,f_L$ at new point sets, and used the same dataset for MLBQ and MLMC to estimate $\Pi[f]$. When using uniform grids, there is no randomness and the experiment is therefore done only once.  The left-hand side plot shows that $\MLBQ[f]$ significantly outperforms $\MLMC[f]$ across a range of budgets $T$. For MLBQ, we also see that the impact of the sample size per level is not as significant as that of type of points used, with the uniform grid outperforming IID points. This is promising since the optimal sample sizes will be difficult to obtain in general due to the need to access $V_l$ or $\|f_l-f_{l-1}\|_{\tau_l}$ for each level $l$ (in the cases of $n^{\text{MLMC}}$ and $n^{\text{MLBQ}}$ respectively). The right-hand side plot shows coverage frequencies for various credible level. Most of the results lie closely to the identity line, indicating that MLBQ has good frequentist coverage. The only exception is for larger budget $T$, in which case MLBQ is under-confident in the sense that the posterior variance is too large relative to frequentist coverage probabilities. This is generally preferable to being over-confident.

\paragraph{ODE with Random Coefficient and Forcing} \label{ODE}

We now consider a popular test-bed for MLMC as first studied in Section 7.1 of \cite{giles2015multilevel}:
\begin{talign*}
    \frac{\mathrm{d}}{\mathrm{d} x}\left( c(x)\frac{\mathrm{d}u}{\mathrm{d} x} \right) &= -50^2 \omega_2^2 \ \text{ for } \ x\in (0,1)
\end{talign*}
with $u(0)=u(1)=0$, $c(x) = 1+\omega_1 x$, $\omega_1 \sim \text{Unif}(0,1)$ and $\omega_2 \sim \mathcal{N}(0,1)$.
The integral is 
\begin{talign*}
    \Pi [f] =\int_\Omega f(\omega) \Pi(d \omega) =\int_\Omega (\int^1_0u(x,\omega)dx) \Pi(d \omega)
\end{talign*}
where $\omega=(\omega_1,\omega_2)$ and $\Pi$ is a product of the marginal distributions for $\omega_1$ and $\omega_2$, and $\Omega=[0,1]\times(-\infty,\infty)$. We take $L=2$ and each level is obtained through a finite difference approximation of $f$ with grid size $h_l$ (see Appendix \ref{appendixODE_solver}). We have $C=(C_0,C_1,C_2) = (1.0, 2.6, 21.8)$ (in $10^{-3}$ seconds).

\begin{figure}[h]
    \centering
    \includegraphics[width=0.48\textwidth]{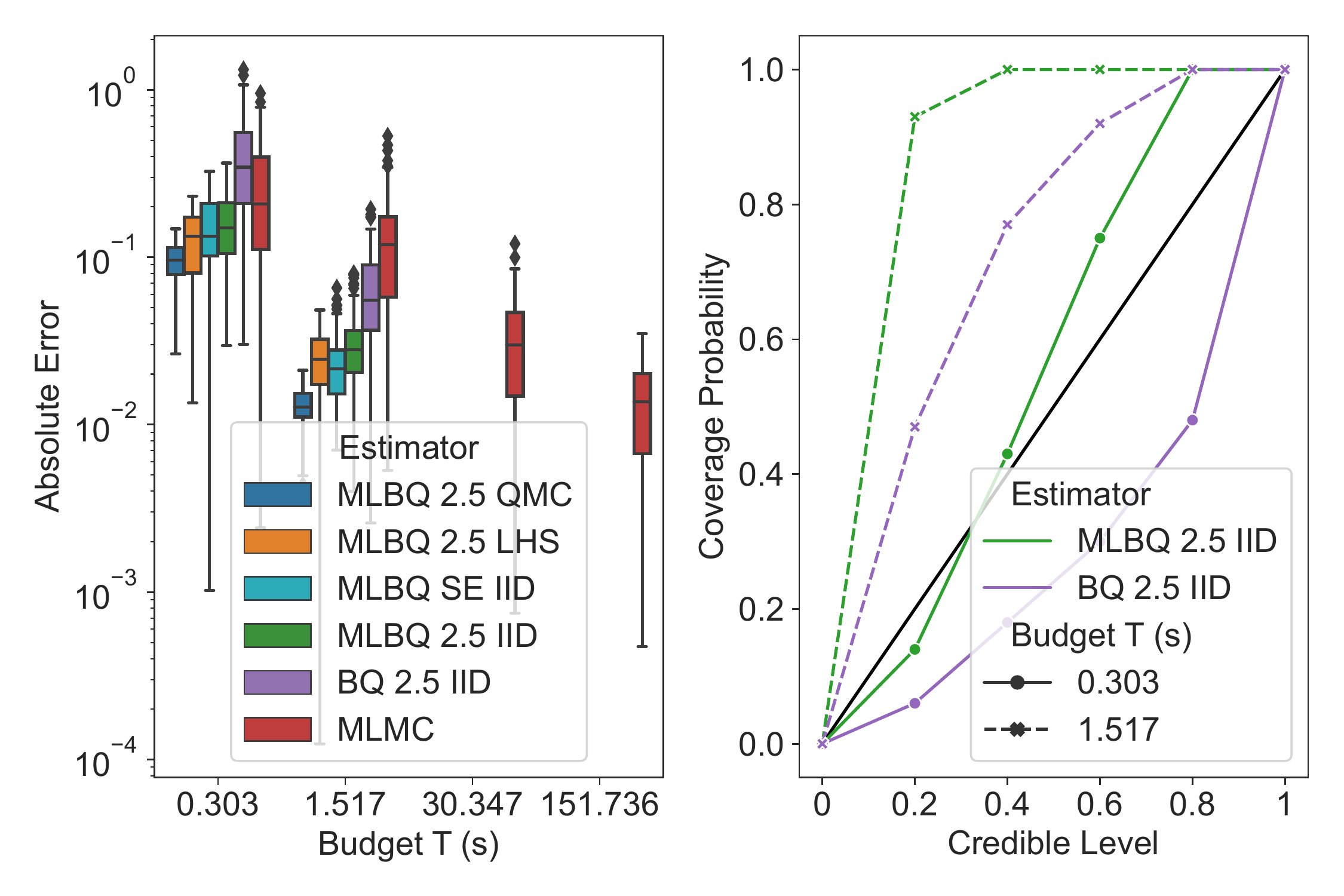}
    \caption{ODE with Random Coefficient and Forcing. \textit{Left}: Absolute integration error. \textit{Right}: Calibration of MLBQ and BQ with IID points.}
    \label{fig:ODE3L}
\end{figure}

The assumptions from \Cref{sec:theory} do not hold here since $\Omega_2$ is unbounded (which breaks A1), but we still use this example to study our method beyond the setting of our theoretical results.
We compare MLBQ with different point sets, MLMC and BQ with IID samples. For all multilevel methods, we select the sample size according to $n^{\text{MLMC}}$ (see Appendix \ref{appendixODE_solver}). In this example, we cannot use $n^{\text{MLBQ}}$ since $\|f_{l}-f_{l-1}\|_{\tau_l}$ is not available in closed form. All methods using a GP with covariance taken to be a product of univariate Mat\'ern  kernels per dimension with $v=2.5$, or a squared exponential kernel (``SE''). 

There are three interesting observations in the left-hand side plot in Figure \ref{fig:ODE3L}. Firstly, MLBQ with a Halton sequence (``QMC'') or a Latin hypercube design (``LHS'') performs better than with IID sampling, once again reflecting the importance of the choice of point set. Secondly, the choice of kernel also has some impact, with the MLBQ estimator with squared exponential kernel outperforming the corresponding estimator with Mat\'ern kernel. Thirdly, MLBQ significantly outperforms BQ and MLMC, even though a sub-optimal sample size per level was used here. More precisely, MLBQ (with any point set) at $T=1.517\text{s}$ is able to outperform MLMC with a budget $20$ times larger ($T=30.347\text{s}$) and is comparable to MLMC with a budget $100$ times larger ( $T=151.736\text{s}$).  A similar conclusion holds when comparing MLBQ with BQ.

Finally, the right-hand side plot shows that the calibration performances of MLBQ and BQ are very similar. The methods tend to be over-confident when $T$ is very small, and become under-confident when $T$ is larger.

\paragraph{Landslide-Generated Tsunami}\label{sec:Tsunami}

We now consider a variation of the submerged landslide-generated tsunami of \cite{lynett2005numerical}. The movement of the landslide mass on the beach slope results in the generation of tsunami waves (see Figure~\ref{fig:sketch}, left), and we consider the temporal evolution of this wave. We use a tsunami simulator called Volna-OP2 \citep{reguly2018volna, giles2020performance}, which is a differential equation solver capable of simulating the complete life-cycle of a tsunami: generation, propagation and inundation. Volna-OP2 is an advanced tsunami simulation tool using unstructured meshes accelerated on GPUs that has been utilised widely by geoscientists, e.g. for real-time tsunami warning systems \citep{giles2021faster} or hazard assessments \citep{gopinathan2021probabilistic, salmanidou2021probabilistic}. Volna-OP2 numerically solves the nonlinear shallow water equations (see Appendix \ref{appendixlandslide}) with a finite volume method. The simulations with Volna-OP2 are run on a single NVIDIA P100 graphical processing unit (the Wilkes2 machine in Cambridge’s CSD3).

\begin{figure}[h]
    \centering
    \includegraphics[trim={0.5cm 0 0 0},clip,scale=0.4]{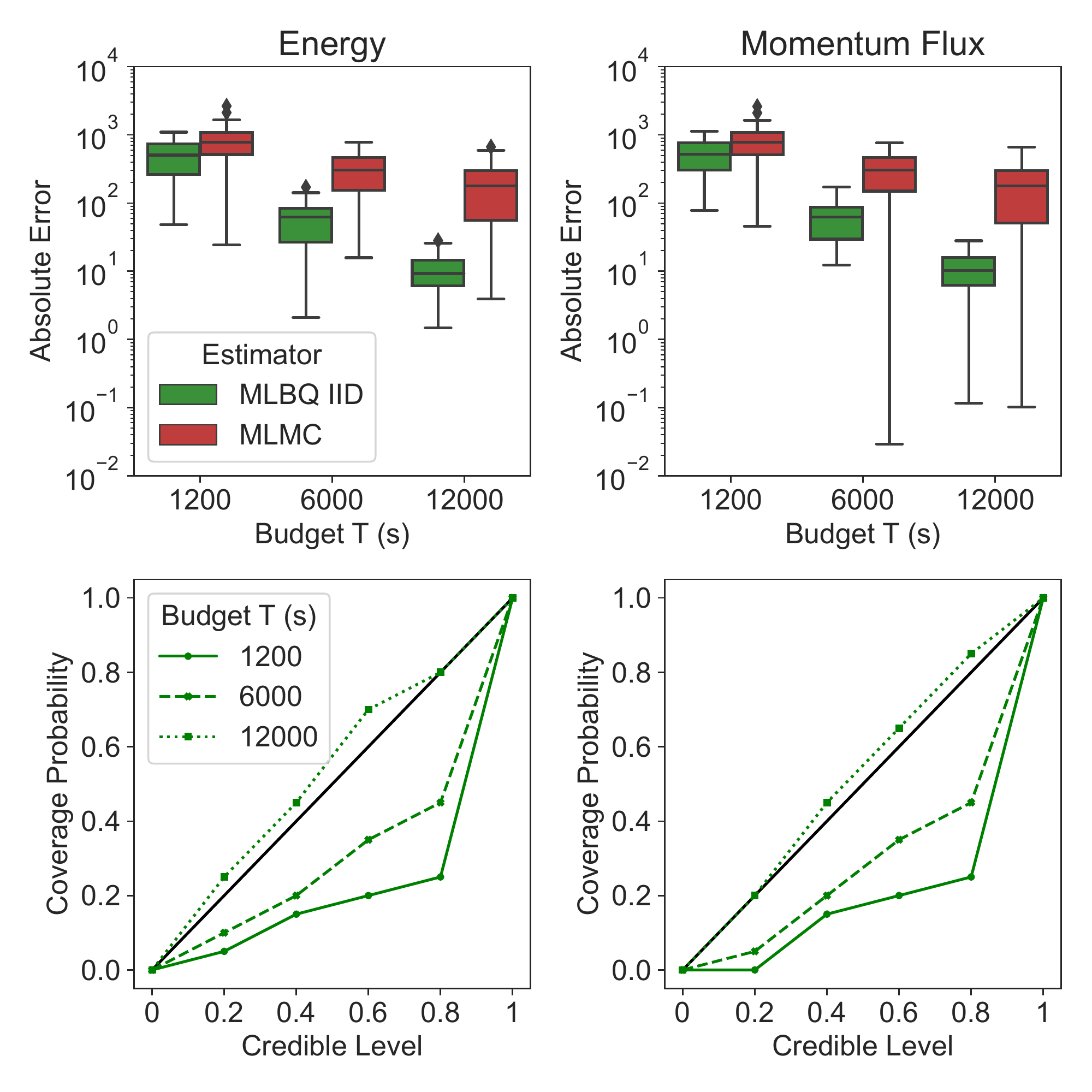}
    \caption{Landslide-Generated Tsunami. \textit{Top}: Absolute integration error.  \textit{Bottom}: Calibration of MLBQ. The left-hand side plots correspond to $\Pi [f^e]$ and the right-hand side plots to $\Pi[f^m]$. }
    \label{TsunamiPlot}
\end{figure}

We use a bathymetry $h(x,t)$ with $x \in [-100, 3100]$ (in meters) and $t \in [0, 300]$ (in seconds). The parameters of interest are:  $\omega_1$, defined to be the ratio of the maximum vertical thickness of the slide ($\Delta h$) to the initial vertical distance from the center point of the slide to the surface ($d_o=50$ m); $\omega_2$, the slope angle; and $\omega_3$, the length of the slide. All of these parameters lead to nonlinear effects which can greatly influence the amplification of tsunami waves. The value of these parameters tends to be unknown a-priori and we take $\Pi$ to consist of marginal distributions representing our uncertainty, given by  $\text{Unif}(0.125, 0.5)$, $\text{Unif}(5^{\circ}, 15^{\circ})$ and $\text{Unif}(100m, 200 m)$ respectively. A representative example of the solution provided by Volna-OP2 for $\omega=(0.375,10^{\circ},150m)$ is on the right-hand side in Figure \ref{fig:sketch}. In tsunami modelling, two functionals of the solution of the model which are often of interest are the total energy flux \citep{degueldre2016random}, denoted $f^e \colon \Omega \rightarrow \R$, and the momentum flux \citep{park2017probabilistic}, denoted $f^m \colon \Omega \rightarrow \R$, and we therefore want to compute $\Pi[f^e]$ and $\Pi[f^m]$.

In the experiments, we estimate these quantities at a gauge at $x=3000$ with MLBQ and MLMC using the same IID point sets and repeat the experiment $20$ times. We take $L=4$ and each level corresponds to a different spatial and temporal resolution used in the solver. The number of evaluations per level are listed in Appendix \ref{appendixlandslide}. We have $C= (C_0,C_1,C_2,C_3,C_4) = (5,15,30,65,150)$ (measured in seconds). These costs are significantly larger than the cost of fitting all GPs, which is carried out on a laptop and ranges from $1$ second to $25$ seconds depending on sample sizes per level. We use a tensor product Mat\'ern kernel with smoothness $v=2.5$ for MLBQ. The related analytical formulas are provided in Appendix \ref{appendix_kmie}.

 The upper box plots of Figure~\ref{TsunamiPlot} show the absolute error of our estimates.  As we observed, MLBQ always significantly outperforms MLMC. More precisely, given a fixed computational budget, MLMC tends to have an error between $10$ and $100$ times smaller than MLBQ. We did not compare to BQ here because $f_L$ is too computational expensive to obtain a reliable estimate. The calibration plots show that MLBQ tends to be overconfident when the budget  is small ($T=1200$ or $T=6000$) but becomes under-confident when budget is larger ($T=12000$). 
 
 Overall, although the computational complexity setup studied in this paper could be considered a `toy model' for the tsunami warning community, any method which showcases such a drastic reduction in computing time could have a signficant impact on tsunami warning centres given their tight budget constraints.

\section{CONCLUSION}

We introduced MLBQ, a method for computing integrals involving multifidelity models. MLBQ enhances MLMC by bringing to it the advantages of Bayesian methods, namely: (1) the ability to make use of prior information about the integrand, which leads to faster convergence rates, and (2) the ability to provide Bayesian quantification of uncertainty over the value of the integral of interest. From the point of view of Bayesian probabilistic numerics, this algorithm is also a step forward towards making the field reach applications where it can be most impactful, including specifically when models are computationally expensive and it is therefore desirable to make use of as much prior knowledge as possible to improve estimates.

There are a large number of possible extensions and we therefore only mention some of the most promising. Firstly, one could consider extending MLBQ to multi-index Monte Carlo \citep{Haji-Ali2016}, which can be useful for models where levels can have multiple indices. For example, in partial differential equation models, one index could be discretisation through time and the other through space, and using this structure could bring further gains. Secondly, one could consider improving scalability through hybrid strategies where BQ is used on the more expensive levels and alternatives, such as MC or scalable BQ methods \citep[e.g.][]{Karvonen2017symmetric,Jagadeeswaran2018}, are used on the cheaper levels. Finally, since we observed that the choice of point set had a large impact on performance, one could consider designing novel acquisition functions for adaptive experimental design \citep[e.g. following the work of][]{Ehara2021}.

\subsubsection*{Acknowledgements}

The authors would like to thank Dimitra Salmadinou for support in accessing tsunami simulations, and Zhuo Sun for helpful discussions. KL and SG acknowledge funding from the Lloyd's Tercentenary Research Foundation, the Lighthill Risk Network and the Lloyd’s Register Foundation-Data Centric Engineering Programme of the Alan Turing Institute for the project ``Future Indonesian Tsunamis: Towards End-to-end Risk quantification (FITTER)''. SG also acknowledges support from The Alan Turing Institute project ``Uncertainty Quantification of multi-scale and multiphysics computer models: applications to hazard and climate models'' under the EPSRC grant EP/N510129/1. DG and SG were supported by the EPSRC project EP/W007711/1 ``Software Environment for Actionable \& VVUQ-evaluated Exascale Application'' (SEAVEA). 
TK was supported by the Academy of Finland postdoctoral researcher grant \#338567 ``Scalable, adaptive and reliable probabilistic integration''. 
FXB was supported by the Lloyd’s Register Foundation Programme on Data-Centric Engineering and The Alan Turing Institute under the EPSRC grant [EP/N510129/1], and through an Amazon Research Award on
“Transfer Learning for Numerical Integration in Expensive Machine Learning Systems”. 

\bibliography{reference}

\appendix
\onecolumn

{
\begin{center}
\LARGE
    \textbf{Supplementary Material}
\end{center}
}

In \Cref{appendix_broader_impact}, we discuss the broader impact of our work. In \Cref{appendix_proof}, we provide the proofs of all theoretical results in the main text. In \Cref{appendix_setup}, we provide details on the experiments introduced in the main text.

\section{BROADER IMPACT}\label{appendix_broader_impact}
Our paper focuses on numerical integration, a common computational problem in statistics and machine learning. Our proposed approach improves the accuracy of approximations and provides probabilistic uncertainty quantification for the value of the integral given limited function evaluations. Our experiments show that applying our approach has the potential to reduce financial and time costs in applications in environmental science and engineering, where high-performance computing clusters are widely used. 

However, it is the specific application that is the determining factor in the broader impact. It depends on whether the user uses the approach for socially beneficial research to have a positive impact. For example, in the landslide tsunami example studied in the experiments section, we showed that the computational requirements can be reduced by using our approach to get useful approximations for tsunami researchers, e.g. designers of tsunami resistant buildings or prevention structures.

\section{PROOFS}\label{appendix_proof}
In this section, we provide the proofs of all theoretical results in the main text. This includes the proof of Proposition \ref{BMLMC_sum} in  \Cref{append:ProofProp1}, the extension in \Cref{append:correlated_case}, the proof of \Cref{theo1} in  \Cref{append:ProofTheo1} and the proof of \Cref{theo2} in \Cref{append:ProofTheo2}. Additionally, for completeness we recall a well-known derivation of  the optimal sample size for MLMC in \Cref{appendix_mlmc_optimN}.

\subsection{Optimal Sample Size for MLMC Given Cost Constraint} \label{appendix_mlmc_optimN}

The optimal sample size $n_0^{\text{MLMC}},\ldots,n_L^{\text{MLMC}}$  that minimize the MSE of MLMC estimates with an overall cost constraint $T$ is the solution to the problem
\begin{talign*}
   n_0^{\text{MLMC}},\ldots,n_L^{\text{MLMC}} \coloneqq  \underset{n_0,n_1,\cdots,n_L}\argmin \sum^L_{l=0} V_l n_l^{-1} \quad \text{s.t.} \quad \sum^L_{l'=0} C_{l'} n_{l'} = T .
\end{talign*}
In this section, we show how the equation above can be solved by using Lagrange multipliers. For some $\lambda>0$, define
\begin{talign*} 
    F_{\text{MLMC}}(n_0,\ldots,n_L,\lambda)=\sum^L_{l=0} V_l n_l^{-1}
    -\lambda \big(T - \sum^L_{l'=0}C_{l'} n_{l'} \big).
\end{talign*}
By taking the derivative of $F_{\text{MLMC}}(n_0,\ldots,n_L,\lambda)$ with respect to $n_0,\ldots,n_L,\lambda$ and setting these equal to $0$, we have
\begin{talign*}
   - V_l n_l^{-2}+\lambda C_l&=0  \quad \Leftrightarrow \quad  n_l= \left(\frac{\lambda C_l}{V_l}\right)^{-\frac{1}{2}} \quad \text{for} \quad l \in \{0,\ldots,L\} \quad \text{and } \quad \sum^L_{l'=0}C_{l'} n_{l'} =T.
\end{talign*}
By plugging the first equation into the second, we have 
\begin{talign*}
    \sum^L_{l'=0}C_{l'} \left(\frac{\lambda C_{l'}}{V_{l'}}\right)^{-\frac{1}{2}}=T 
    \quad\Leftrightarrow \quad 
   \lambda^{-\frac{1}{2}} =\frac{T}{\sum^L_{l'=0}C_{l'} \left(\frac{C_{l'}}{V_{l'}}\right)^{-\frac{1}{2}}}
   \quad\Leftrightarrow \quad
   \lambda= \left(\frac{1}{T} \sum^L_{l'=0}C_{l'} \left(\frac{C_{l'}}{V_{l'}}\right)^{-\frac{1}{2}}\right)^2.
\end{talign*}
Finally, plugging this expression for $\lambda$ into our expression for $n_l$, we get
\begin{talign*}
    n_l^{\text{MLMC}} = T \sqrt{\frac{V_l}{C_l}} \left(\sum^L_{l^\prime=0} \sqrt{V_{l'} C_{l'}}\right)^{-1} \qquad \text{for} \qquad l \in \{0,\ldots,L\}.
\end{talign*}

\subsection{Proof of Proposition~\ref{BMLMC_sum}}\label{append:ProofProp1}

\begin{proof}[Proof of Proposition~\ref{BMLMC_sum}]
If $f_l-f_{l-1}$ for $l \in \{ 0,\ldots,L\}$ are a-priori independent, and  $f_l-f_{l-1} \sim \mathcal{GP}(m_l,c_l)$, then
 \begin{talign*}
   \begin{bmatrix}
        f_0(W_0)\\
         f_1(W_1)-f_0(W_1)\\
         \vdots\\
         f_L(W_L)-f_{L-1}(W_L)\\
         \sum^L_{l=0}\left(f_l(W_*) -f_{l-1}(W_*)\right)
    \end{bmatrix}
    \sim
   &   \mathcal{N}
    \left( 
      \begin{bmatrix}
         m_0(W_0)\\
         m_1(W_1)\\
          \vdots\\
         m_L(W_L)\\
         \sum_{l=0}^L m_l(W_*)
     \end{bmatrix},  \right.\\
  & \left. \qquad \qquad  \begin{bmatrix}
         c_0(W_0,W_0) & 0 & \cdots & 0 & c_0(W_0,W_*)\\
         0 & c_1(W_1,W_1) & \cdots & 0 & c_1(W_1,W_*)\\
         \vdots & \vdots &  \ddots & \vdots  &  \vdots \\
         0 & 0 & \cdots & c_L(W_L,W_L) & c_L(W_L,W_*)\\
         c_0(W_*,W_0) & c_1(W_*,W_1) &  \cdots & c_L(W_*,W_L) & \sum_{l=0}^L c_l(W_*,W_*)
    \end{bmatrix}
    \right ),
\end{talign*}
where $W_*=(\omega_{(*,1)},\ldots,\omega_{(*,n_*)})^\top$ are query locations. Applying the formula for the conditional distribution from the multivariate Gaussian distribution, given the vector values $\left( \right. f_0(W_0), f_1(W_1)-f_0(W_1), \cdots, f_L(W_L)-f_{L-1}(W_L)\left. \right)^\top$, the conditional distribution of $\sum_{l=0}^{L}\left(f_l(W_*)-f_{l-1}(W_*)\right)$ is Gaussian, with mean 
\begin{talign*}
&\sum_{l=0}^L m_l(W_*)+\begin{bmatrix}
          c_0(W_*,W_0) & c_1(W_*,W_1) &  \cdots & c_L(W_*,W_L)
    \end{bmatrix} \\
  & \hspace{3cm}   \times
    \begin{bmatrix}
         c_0(W_0,W_0) & 0 & \cdots & 0 \\
         0 & c_1(W_1,W_1) & \cdots & 0 \\
         \vdots & \vdots &  \ddots & \vdots \\
         0 & 0 & \cdots & c_L(W_L,W_L) \\
    \end{bmatrix}^{-1}
\begin{bmatrix}
         f_0(W_0)-m_0(W_0)\\
         f_1(W_1)-f_0(W_1)-m_1(W_1)\\
         \vdots\\
         f_L(W_L)-f_{L-1}(W_L)-m_L(W_L)
    \end{bmatrix} \\
& \qquad=\sum_{l=0}^L \left(m_l(W_*) + c_l(W_*,W_l)c_l(W_l,W_l)^{-1}\left(f_l(W_l)-f_{l-1}(W_l)-m_l(W_l)\right) \right)\\
& \qquad=\sum_{l=0}^L \tilde{m}_l(W_*),
\end{talign*}
and variance
\begin{talign*}
 &\sum^L_{l=0}c_l(W_*,W_*)+\begin{bmatrix}
          c_0(W_*,W_0) & c_1(W_*,W_1) &  \cdots & c_L(W_*,W_L)
    \end{bmatrix} \\
 & \hspace{4.5cm} \times
    \begin{bmatrix}
         c_0(W_0,W_0) & 0 & \cdots & 0 \\
         0 & c_1(W_1,W_1) & \cdots & 0 \\
         \vdots & \vdots &  \ddots & \vdots \\
         0 & 0 & \cdots & c_L(W_L,W_L) \\
    \end{bmatrix}^{-1}
 \begin{bmatrix}
         c_0(W_0,W_*)\\
         c_1(W_1,W_*)\\
         \vdots\\
       c_L(W_L,W_*)
    \end{bmatrix}\\
& \qquad=\sum_{l=0}^L \left(c_l(W_*,W_*) +  c_l(W_*,W_l)c_l(W_l,W_l)^{-1}c_l(W_l,W_*)\right) \\
& \qquad=\sum_{l=0}^L \tilde{c}_l(W_*,W_*).
\end{talign*}
As a result, the posterior on $\sum_{l=0}^L \left(f_l-f_{l-1}\right)$ is
$\mathcal{GP}(\sum_{l=0}^L \tilde{m}_l,\sum_{l=0}^L \tilde{c}_l)$. The posterior on $\sum_{l=0}^L\Pi[f_l-f_{l-1}]$ can be obtained (following the usual derivation for the BQ distribution on integrals) integrating the posterior mean and covariance functions and takes the form of a univariate Gaussian with mean
\begin{talign*}
   \E_{\textup{MLBQ}}[\Pi[f]] 
     &  = \sum_{l=0}^L \big( \Pi[m_l] + \Pi[c_l(\cdot,W_l)]c_l(W_l,W_l)^{-1}  (f_l(W_l)-f_{l-1}(W_l)-m_l(W_l)) \big)\\
      &= \sum_{l=0}^L \BQ[f_l -f_{l-1} ]
\end{talign*}
and variance
\begin{talign*}
    \V_{\textup{MLBQ}}[\Pi[f]] 
   & =  \sum_{l=0}^L  \big( \Pi[\Pi[c_l]]  -  \Pi[c_l(\cdot,W_l)]c_l(W_l,W_l)^{-1}\Pi[c_l(W_l,\cdot)] \big)\\
   &= \sum^{L}_{l=0}\mathbb{V}_{\textup{BQ}}[\Pi[f_l-f_{l-1}]].
\end{talign*}
\end{proof}

\subsection{Extension of Proposition~\ref{BMLMC_sum}}\label{append:correlated_case}
For the vector-valued function $(f_0,f_1-f_0,\cdots,f_L -f_{L-1})^\top$,  suppose we specify a separable kernel $C(\omega,\omega') = B c(\omega,\omega')$, where $B\in \mathbb{R}^{L\times L}$ is symmetric and positive definite with $B_{i,j}$ denoting the $(i,j)$-entry of $B$, then
 \begin{talign*}
   & \begin{bmatrix}
        f_0(W_0)\\
         f_1(W_1)-f_0(W_1)\\
         \vdots\\
         f_L(W_L)-f_{L-1}(W_L)\\
         \sum^L_{l=0}\left(f_l(W_*) -f_{l-1}(W_*)\right)
    \end{bmatrix}
    \sim
    \mathcal{N}
    \left( 
        \begin{bmatrix}
         m_0(W_0)\\
         m_1(W_1)\\
          \vdots\\
         m_L(W_L)\\
         \sum_{l=0}^L m_l(W_*)
    \end{bmatrix}, \right. \\
  &\quad \qquad \left.  \begin{bmatrix}
         B_{0,0}c(W_0,W_0) & B_{0,1}c(W_0,W_1) & \cdots & B_{0,L}c(W_0,W_L) & \sum^L_{l=0}B_{0,l}c(W_0,W_*)\\
         B_{1,0}c(W_1,W_0) & B_{1,1}c(W_1,W_1) & \cdots & B_{1,L}c(W_1,W_L) & \sum^L_{l=0}B_{1,l}c(W_1,W_*)\\
         \vdots & \vdots &  \ddots &  \vdots &  \vdots \\
         B_{L,0}c(W_L,W_0) & B_{L,1}c(W_L,W_1) & \cdots & B_{L,L}c(W_L,W_L) & \sum^L_{l=0}B_{L,l}c(W_L,W_*)\\
         \sum^L_{l=0}B_{l,0}c(W_*,W_0) & \sum^L_{l=0}B_{l,1}c(W_*,W_1) &  \cdots & \sum^L_{l=0}B_{l,L}c(W_*,W_L) & \sum^L_{l=0}\sum^L_{l'=0}B_{l,l'}c(W_*,W_*)
    \end{bmatrix}
    \right ).
\end{talign*}
Similarly, applying the formula for Gaussian conditionals, given the vector values $\left( \right. f_0(W_0), f_1(W_1)-f_0(W_1), \cdots, f_L(W_L)-f_{L-1}(W_L)\left. \right)^\top $, the conditional distribution of $\sum_{l=0}^{L}\left(f_l(W_*)-f_{l-1}(W_*)\right)$ is Gaussian, with mean 
\begin{talign*}
  \tilde{m}^B(W_*)=&\sum_{l=0}^L m_l(W_*)+\begin{bmatrix}
          \sum^L_{l=0}B_{l,0}c(W_*,W_0) & \sum^L_{l=0}B_{l,1}c(W_*,W_1) &  \cdots & \sum^L_{l=0}B_{l,L}c(W_*,W_L)
    \end{bmatrix} \\
 &  \hspace{0.1cm}  \times
    \begin{bmatrix}
         B_{0,0}c(W_0,W_0) & B_{0,1}c(W_0,W_1) & \cdots & B_{0,L}c(W_0,W_L) \\
         B_{1,0}c(W_1,W_0) & B_{1,1}c(W_1,W_1) & \cdots & B_{1,L}c(W_1,W_L) \\
         \vdots & \vdots &  \ddots &  \vdots  \\
         B_{L,0}c(W_L,W_0) & B_{L,1}c(W_L,W_1) & \cdots & B_{L,L}c(W_L,W_L) 
    \end{bmatrix}^{-1}
 \begin{bmatrix}
         f_0(W_0)-m_0(W_0)\\
         f_1(W_1)-f_0(W_1)-m_1(W_1)\\
         \vdots\\
         f_L(W_L)-f_{L-1}(W_L)-m_L(W_L)
    \end{bmatrix}
\end{talign*}
and variance
\begin{talign*}
\tilde{c}^B(W_*,W_*)= &\sum^L_{l=0}\sum^L_{l'=0}B_{l,l'}c(W_*,W_*)\\
&\hspace{0.3cm}+\begin{bmatrix}
    \sum^L_{l=0}B_{l,0}c(W_*,W_0) & \sum^L_{l=0}B_{l,1}c(W_*,W_1) & \cdots & \sum^L_{l=0}B_{l,L}c(W_*,W_L)
    \end{bmatrix} \\
  & \qquad  \times
   \begin{bmatrix}
         B_{0,0}c(W_0,W_0) & B_{0,1}c(W_0,W_1) & \cdots & B_{0,L}c(W_0,W_L) \\
         B_{1,0}c(W_1,W_0) & B_{1,1}c(W_1,W_1) & \cdots & B_{1,L}c(W_1,W_L) \\
         \vdots & \vdots & \ddots &  \vdots  \\
         B_{L,0}c(W_L,W_0) & B_{L,1}c(W_L,W_1) & \cdots & B_{L,L}c(W_L,W_L) 
    \end{bmatrix}^{-1}
 \begin{bmatrix}
         \sum^L_{l=0}B_{0,l}c(W_0,W_*)\\
         \sum^L_{l=0}B_{1,l}c(W_1,W_*)\\
         \vdots\\
        \sum^L_{l=0}B_{L,l}c(W_L,W_*)
    \end{bmatrix}.
\end{talign*}
As a result, the posterior on $\sum_{l=0}^L \left(f_l-f_{l-1}\right)$ is $\mathcal{GP}(\tilde{m}^B,\tilde{c}^B)$. Similarly, the posterior on $\sum_{l=0}^L \Pi[f_l-f_{l-1}]$ can be obtained (following the usual derivation for the BQ distribution on integrals) integrating the posterior mean and covariance functions and takes the form of a univariate Gaussian with mean
\begin{talign*}
  \sum_{l=0}^L\Pi[m_l]+&\begin{bmatrix}
          \sum^L_{l=0}B_{l,0}\Pi[c(\cdot,W_0)] & \sum^L_{l=0}B_{l,1}\Pi[c(\cdot,W_1)] &  \cdots & \sum^L_{l=0}B_{l,L}\Pi[c(\cdot,W_L)]
    \end{bmatrix} \\
  & \hspace{-0.1cm} \times
    \begin{bmatrix}
         B_{0,0}c(W_0,W_0) & B_{0,1}c(W_0,W_1) & \cdots & B_{0,L}c(W_0,W_L) \\
         B_{1,0}c(W_1,W_0) & B_{1,1}c(W_1,W_1) & \cdots & B_{1,L}c(W_1,W_L) \\
         \vdots & \vdots &  \ddots &  \vdots  \\
         B_{L,0}c(W_L,W_0) & B_{L,1}c(W_L,W_1) & \cdots & B_{L,L}c(W_L,W_L) 
    \end{bmatrix}^{-1}
 \begin{bmatrix}
         f_0(W_0)-m_0(W_0)\\
         f_1(W_1)-f_0(W_1)-m_1(W_1)\\
         \vdots\\
         f_L(W_L)-f_{L-1}(W_L)-m_L(W_L)
    \end{bmatrix}
\end{talign*}
and variance
\begin{talign*}
  \sum^L_{l=0}\sum^L_{l'=0}B_{l,l'}\Pi[\Pi[c]]+&\begin{bmatrix}
          \sum^L_{l=0}B_{l,0}\Pi[c(\cdot,W_0)] & \sum^L_{l=0}B_{l,1}\Pi[c(\cdot,W_1)] &  \cdots & \sum^L_{l=0}B_{l,L}\Pi[c(\cdot,W_L)]
    \end{bmatrix} \\
 &\hspace{-0.2cm} \times
    \begin{bmatrix}
         B_{0,0}c(W_0,W_0) & B_{0,1}c(W_0,W_1) & \cdots & B_{0,L}c(W_0,W_L) \\
         B_{1,0}c(W_1,W_0) & B_{1,1}c(W_1,W_1) & \cdots & B_{1,L}c(W_1,W_L) \\
         \vdots & \vdots &  \ddots &  \vdots  \\
         B_{L,0}c(W_L,W_0) & B_{L,1}c(W_L,W_1) & \cdots & B_{L,L}c(W_L,W_L) 
    \end{bmatrix}^{-1}
 \begin{bmatrix}
         \sum^L_{l=0}B_{0,l}\Pi[c(W_0,\cdot)]\\
         \sum^L_{l=0}B_{1,l}\Pi[c(W_1,\cdot)]\\
         \vdots\\
        \sum^L_{l=0}B_{L,l}\Pi[c(W_L,\cdot)]
    \end{bmatrix}.
\end{talign*}

An example is provided in Appendix \ref{append: Experimental_Details_PE}.

\subsection{Proof of Theorem~\ref{theo1}} \label{append:ProofTheo1}

\begin{proof}[Proof of Theorem~\ref{theo1}]
Suppose that $c$ is a covariance function such that $\mathcal{H}(c)$ is norm-equivalent to $W_2^\alpha(\Omega)$ and $f \in W_2^\beta(\Omega)$ for $\alpha \geq \beta > d/2$.
Since the density $\pi$ of $\Pi$ is bounded by Assumption~A2, we have
\begin{talign*}
  \lvert \Pi[f] - \BQ[f] \rvert = \big\lvert \int_\Omega f(\omega) \pi(\omega) d \omega - \int_\Omega \tilde{m}(\omega) \pi(\omega) d \omega \big\rvert \leq \| \pi \|_{L^\infty(\Omega)} \int_\Omega \lvert f(\omega) - \tilde{m}(\omega) \rvert d \omega,
\end{talign*}
where $\tilde{m}$ is the GP posterior mean given observations of $f$ at $n$ points $W = (\omega_1, \ldots, \omega_n)^\top$ and with $m \equiv 0$.
Because $c(\cdot, \omega) \in \mathcal{H}(c) = W_2^\alpha(\Omega) \subset W_2^\beta(\Omega)$, the posterior mean is an element of $W_2^\beta(\Omega)$.
Assumption~A1 ensures that $\Omega$ satisfies the assumptions of Theorem~4.1 in~\citet{Arcangeli2007}.
Using this theorem with $p = 2$, $q = 1$, $l = 0$, $r = \beta$, and $n = d$ gives
\begin{talign*}
  \int_\Omega \lvert f(\omega) - \tilde{m}(\omega) \rvert d \omega \leq \tilde{a} h_{W, \Omega}^{\beta} \| f - \tilde{m} \|_{\beta}
\end{talign*}
whenever $h_{W, \Omega}$ is sufficiently small.
The positive constant $\tilde{a}$ depends only on $\alpha$, $\beta$, $d$, and $\Omega$.
Theorem~4.2 in \citet{Narcowich2006} with $\mu = \beta$ and $\tau = \alpha$ and the well known identification of the GP posterior mean with the minimum-norm kernel interpolant~\citep[e.g.,][Section~3]{kanagawa2018gaussian} yield
\begin{talign*}
  \| f - \tilde{m} \|_{\beta} \leq a' \rho_{W, \Omega}^{\alpha - \beta} \| f \|_{\beta},
\end{talign*}
where $\rho_{W, \Omega} = h_{W, \Omega} / q_X$ for $q_X = \frac{1}{2} \min_{i \neq j} \| \omega_i - \omega_j \|$ is known as the mesh ratio and $a'$ is a positive constant that does not depend on $f$.
Assume that the points $W$ are quasi-uniform with the constant $h_\textup{qu} > 0$, which is to say that $h_{W, \Omega} \leq h_\textup{qu} n^{-1/d}$.
This implies that $\rho_{W, \Omega} \leq \tilde{h}_\textup{qu}$ for a constant $\tilde{h}_\textup{qu}$ that depends on $h_\textup{qu}$, $d$, and $\Omega$.
Combining all the bounds above yields
\begin{talign} \label{eq:theo1-proof-bound-1}
  \lvert \Pi[f] - \BQ[f] \rvert \leq \| \pi \|_{L^\infty(\Omega)} \tilde{a} h_{W, \Omega}^{\beta} \| f - \tilde{m} \|_{\beta} \leq \| \pi \|_{L^\infty(\Omega)} \tilde{a} h_\textup{qu}^\beta n^{-\beta/d} a' \tilde{h}_\textup{qu}^{\alpha - \beta} \| f \|_{\beta} \eqqcolon \| \pi \|_{L^\infty(\Omega)} a \| f \|_{\beta} n^{-\beta / d}
\end{talign}
for a positive constant $a$ that depends only on $\alpha$, $\beta$, $c$, $h_\textup{qu}$, $d$, and $\Omega$.
If $f \in W_2^\beta(\Omega)$ for $\beta > \alpha > d/2$, we use the fact that in this case $f \in W_2^\alpha(\Omega)$ and obtain~\eqref{eq:theo1-proof-bound-1} with $\beta = \alpha$.
Let $\tau = \min\{\alpha, \beta\}$.
We can then write
\begin{talign} \label{eq:theo1-proof-bound-2}
  \lvert \Pi[f] - \BQ[f] \rvert \leq \|\pi \|_{L^\infty(\Omega)}a \| f \|_{\tau} n^{-\tau / d}
\end{talign}
for any reals $\alpha$ and $\beta$ that exceed $d/2$.

We then apply~\eqref{eq:theo1-proof-bound-2} to each term in MLBQ error under Assumptions~A1--A6 (recall the convention $f_{-1} \equiv 0$):
\begin{talign*}
  \Err(\MLBQ) = \lvert \Pi[f_L]-\MLBQ[f_L] \rvert &= \big\lvert \Pi[f_0]+\sum^L_{l=1}\Pi[f_l-f_{l-1}]-\BQ[f_0]-\sum^L_{l=1}\BQ[f_l-f_{l-1}] \big\rvert \\
  &\leq \sum_{l=0}^L \big\lvert \Pi[f_l - f_{l-1}]-\BQ[f_l - f_{l-1}] \big\rvert \\
  &\leq \| \pi \|_{L^\infty(\Omega)} \sum_{l=0}^L a_l \|f_l-f_{l-1}\|_{\tau_l}n_l^{-\tau_l/d},
\end{talign*}
where Assumption~A4 ensures that $f_l - f_{l-1} \in W_2^{\beta_l}(\Omega) \subset W_2^{\tau_l}(\Omega)$.
\end{proof}


\subsection{Proof of Theorem \ref{theo2}} \label{append:ProofTheo2}

\begin{proof}[Proof of Theorem \ref{theo2}]
The sample sizes $n_0^{\text{MLBQ}},\ldots,n_L^{\text{MLBQ}}$  that minimize the upper bound of the absolute error of MLBQ in \Cref{theo1} with the overall cost constraint $T$ are
\begin{talign*}
    n_0^{\text{MLBQ}}, \ldots, n_L^{\text{MLBQ}} := \underset{n_0,n_1,\cdots,n_L}\argmin \sum^L_{l=0} A_l n_l^{-\frac{\tau}{d}} ~~ \quad \text{s.t.} \quad \gamma \sum^L_{l'=0}C_{l'} n_{l'} = T,
\end{talign*}
where $A_l = a \|f_l-f_{l-1}\|_{\tau}$. We note that the $\| \pi \|_{L^\infty(\Omega)} $ term can be ignored since it does not depend on the sample sizes. Similarly to the derivation for MLMC in \Cref{appendix_mlmc_optimN}, the optimisation problem above can be solved by using Lagrange multipliers. For some $\lambda>0$, we define
\begin{talign} \label{Fcase1}
    F_{\text{MLBQ}}(n_0,\ldots,n_L,\lambda)=\sum^L_{l=0} A_l n_l^{-\frac{\tau}{d}}-\lambda\big( T-\gamma\sum^L_{l'=0}C_{l'} n_{l'}\big).
\end{talign}
Differentiating $F_{\text{MLBQ}}(n_0,\ldots,n_L,\lambda)$ with respect to $n_0,\ldots,n_L,\lambda$ and setting the equations equal to $0$ gives
\begin{talign*}
   - \frac{\tau}{d} A_l n_l^{-\frac{\tau}{d}-1}+\lambda \gamma C_l=0 
   \quad \Leftrightarrow \quad
   n_l = \left( \frac{ d \lambda \gamma C_l}{\tau A_l} \right)^{-\frac{d}{\tau + d}} \quad \text{for} \quad  l \in \{0,\ldots,L\}
   \quad \text{and } \quad
    \gamma \sum^L_{l'=0}C_{l'} n_{l'} = T.
\end{talign*}
By plugging the first equation into the second, we get
\begin{talign*}
    \sum^L_{l'=0} \gamma C_{l'} \left( \frac{ d \lambda \gamma C_{l'}}{\tau A_{l'}} \right)^{-\frac{d}{\tau + d}} = T \qquad 
\Leftrightarrow \qquad
    \lambda= T^{-\frac{\tau + d}{d}}
    \left( \sum^L_{l'=0} \gamma C_{l'} \left( \frac{d\gamma C_{l'}}{\tau A_{l'}} \right)^{-\frac{d}{\tau + d}} \right)^{\frac{\tau + d}{d}}.  
\end{talign*}
Plugging this last expression for $\lambda$ into our expression for $n_l$, we get
\begin{talign*}
  n_l^{\text{MLBQ}} 
  = & \frac{T}{\gamma} 
  \left( \frac{C_l}{A_l} \right)^{-\frac{d}{\tau + d}}
  \left( \sum^L_{{l^\prime}=0}C_{l^\prime} \left( \frac{C_{l^\prime}}{A_{l^\prime}} \right)^{-\frac{d}{\tau + d}} \right)^{-1}\\
  = &\frac{T}{\gamma} \left( \frac{ \|f_l-f_{l-1}\|_{\tau}}{C_l} \right)^{\frac{d}{\tau + d}} \left( \sum^L_{{l'}=0}C_{l'}^{\frac{\tau}{\tau + d}} \left( \|f_{l'}-f_{{l'}-1}\|_{\tau} \right)^{\frac{d}{\tau + d}} \right)^{-1} \quad \text{for} \quad l \in \{0,\ldots,L\}. 
\end{talign*}
\end{proof}


\section{ADDITIONAL EXPERIMENTS}\label{appendix_setup}
In this section, we provide details of the experimental setup and additional experiments. This includes details of the three experiments in main text in \Cref{append: Experimental_Details_PE}, \Cref{appendixODE_solver} and \Cref{appendixlandslide}, additional experiments in \Cref{append: Experimental_Details_PE}, \Cref{appendixstep} and \Cref{appendixME} and the analytical formulae for the kernel mean and initial error in \Cref{appendix_kmie}.

\subsection{Experiment 1: Poisson Equation} \label{append: Experimental_Details_PE}

\paragraph{Construction of the Levels} \label{Append: FEM}
Given the specific example of Poisson equation under consideration, we are able to obtain a closed form solution to the PDE: $f(\omega) = \frac{1}{2} \omega (\omega - 1)$. We construct a piecewise linear finite element approximation $f_l$ of the solution $f$ on level $l$ as follows.
Let $p_l \in \mathbb{N}$ and $0 < \omega_{l,1} < \cdots < \omega_{l, p_l} < 1$.
Define the piecewise linear finite element basis functions as
\begin{talign*}
  v_{l, i}(\omega) = 
  \begin{cases}
    \frac{\omega - \omega_{l, i-1}}{\omega_{l,i} - \omega_{l, i-1}} \: & \text{ if } \: \omega \in [\omega_{l,i-1}, \omega_{l,i}], \\
    \frac{\omega_{l,i+1} - \omega}{\omega_{l,i+1} - \omega_{l, i}} \: & \text{ if } \: \omega \in [\omega_{l,i}, \omega_{l,i+1}], \\
    0 \: & \text{ otherwise}.
  \end{cases}
\end{talign*}
The $i$th basis function is supported on $[\omega_{l,i-1}, \omega_{l, i+1}]$.
Here we use the conventions $\omega_{l,0} = 0$ and $\omega_{l,p_l+1} = 1$.
The finite element approximation $f_l$ to $f$ is given by $f_l(\omega) = \sum_{i=1}^{p_l} a_{l,i} v_{l,i}(\omega)$, where the coefficient vector $a_l = (a_{l,1}, \ldots, a_{l,p_l})^\top \in \mathbb{R}^{p_l}$ is solved from the linear system $ -L_l a_l = g_l,$ where $L_l \in \mathbb{R}^{p_l \times p_l}$ is the tridiagonal stiffness matrix with 
\begin{talign*}
  (L_l)_{i,i} & = \int_0^1 v_{l,i}'(\omega)^2 d \omega = \frac{1}{\omega_{l,i} - \omega_{l,i-1}} + \frac{1}{\omega_{l,i+1} - \omega_{l,i}}, \qquad
  & (L_l)_{i,i-1} = (L_l)_{i-1,i} = \int_0^1 v_{l,i}'(\omega) v_{l,j}'(\omega) d\omega = -\frac{1}{\omega_{l,i} - \omega_{l, i-1}},
\end{talign*}
and the vector $g_l \in \mathbb{R}^{p_l}$ has elements $(g_l)_i = \int_0^1 f(\omega) v_{l,i}(\omega) d\omega$. Consider now the Brownian motion kernel $c_l(\omega, \omega') = \sigma_l^2 c_\textup{BM}(\omega, \omega') = \sigma_l^2 \min\{\omega, \omega'\}$, for a positive amplitude parameter $\sigma_l$, where the RKHS of the Brownian motion kernel $c_l(\omega, \omega')$  on $\Omega  = [0, 1]$  consists of functions $g_l(0) = 0$ and $f_l \in  W_2^1([0, 1])$ \citep{karvonen2020maximum}. It is straightforward to verify that a piecewise linear finite element basis function can be written in terms of the Brownian motion kernel translates:
\begin{talign*}
  \begin{split}
    v_{l,i}(\omega) ={}& -\frac{1}{\omega_{l,i} - \omega_{l,i-1}} c_\textup{BM}(\omega, \omega_{l, i-1}) + \left( \frac{1}{\omega_{l,i} - \omega_{l,i-1}} + \frac{1}{\omega_{l,i+1} - \omega_{l,i}} \right) c_\textup{BM}(\omega, \omega_{l, i}) - \frac{1}{\omega_{l,i+1} - \omega_{l,i}} c_\textup{BM}(\omega, \omega_{l, i+1})
    \end{split}
\end{talign*}
and the full finite element approximation is
\begin{talign*}
  \begin{split}
    f_l &(\omega) ={} \sum_{i=1}^{p_l} a_{l,i} v_{l,i}(\omega) \\
    ={}&
    \sum_{i=1}^{p_l} a_{l,i} \left[ -\frac{1}{\omega_{l,i} - \omega_{l,i-1}} c_\textup{BM}(\omega, \omega_{l, i-1}) + \left( \frac{1}{\omega_{l,i} - \omega_{l,i-1}} + \frac{1}{\omega_{l,i+1} - \omega_{l,i}} \right) c_\textup{BM}(\omega, \omega_{l, i}) - \frac{1}{\omega_{l,i+1} - \omega_{l,i}} c_\textup{BM}(\omega, \omega_{l, i+1}) \right] \\
    ={}&
    \left( \frac{a_{l,1}}{\omega_{l,1}} + \frac{a_{l,1} - a_{l,2}}{\omega_{l,2} - \omega_{l,1}} \right) c_\textup{BM}(\omega, \omega_{l, 1}) 
    + \left( \frac{a_{l,p_l}}{1 - \omega_{l,p_l}} + \frac{a_{l,p_l} - a_{l,p_l-1}}{\omega_{l,p_l} - \omega_{l,p_l-1}} \right) c_\textup{BM}(\omega, \omega_{l, p_l}) 
    - \frac{a_{l,p_l}}{1 - \omega_{l,p_l}} c_\textup{BM}(\omega, 1) \\
    &
    + \sum_{i=2}^{p_l-1} \left( \frac{a_{l,i} - a_{l,i-1}}{\omega_{l,i} - \omega_{l,i-1}} + \frac{a_{l,i} - a_{l,i+1}}{\omega_{l,i+1} - \omega_{l,i}}\right) c_\textup{BM}(\omega, \omega_{l, i}) ,
    \end{split}
\end{talign*}
where we have used the fact that $c_\textup{BM}(\omega, \omega_0) = c_\textup{BM}(\omega, 0) = 0$ for all $\omega \geq 0$ and $\omega_{l,p_l+1} = 1$.
For simplicity, suppose that the points are equispaced on $[0, 1]$ so that $\Delta_l = \omega_{l,1} = \omega_{l,i} - \omega_{l,i-1} = 1 - \omega_{l,p_l}$  for every $i = 1,\ldots,p_l$. Then the finite element approximation simplifies to
\begin{talign*}
  \begin{split}
    f_l(\omega) ={}& \frac{2a_{l,1} - a_{l,2}}{\Delta_l} c_\textup{BM}(\omega, \omega_{l, 1}) + \frac{2a_{l,p_l} - a_{l,p_l-1}}{\Delta_l} c_\textup{BM}(\omega, \omega_{l, p_l}) - \frac{a_{l,p_l}}{\Delta_l } c_\textup{BM}(\omega, 1) \\
    &+ \frac{1}{\Delta_l} \sum_{i=2}^{p_l-1} ( 2a_{l,i} - a_{l,i-1} - a_{l,i+1} ) c_\textup{BM}(\omega, \omega_{l, i}) \\
    ={}& 
    \frac{1}{\Delta_l} \left[ -a_{l,p_l} c_\textup{BM}(\omega, 1) + \sum_{i=1}^{p_l} ( 2a_{l,i} - a_{l,i-1} - a_{l,i+1} ) c_\textup{BM}(\omega, \omega_{l, i}) \right] \\
    ={}& 
    \frac{1}{\Delta_l \sigma_l^2} \left[ -a_{l,p_l} c_l(\omega, 1) + \sum_{i=1}^{p_l} ( 2a_{l,i} - a_{l,i-1} - a_{l,i+1} ) c_l(\omega, \omega_{l, i}) \right]
    \end{split}
\end{talign*}
where we use the convention $a_0 = a_{l,p_l+1} = 0$.
Denote $b_{l,i} = 2a_{l,i} - a_{l,i-1} - a_{l,i+1}$.
Using the above expression for $f_l$ as a sum of kernel translates and the general formula $\left\| \sum_{i=1}^n \alpha_i c(\cdot, \omega_i) \right\|_\mathcal{H}^2 = \sum_{i=1}^n \sum_{j=1}^n \alpha_i \alpha_j c(\omega_i, \omega_j)$
we are able to compute the squared RKHS norm:
\begin{talign*}
  \begin{split}
  \| f_l \|_{\mathcal{H}_l}^2 
  &= 
  \frac{1}{\Delta_l^2} \left[ a_{l,p_l}^2 c_\textup{BM}(1, 1) -2 a_{l,p_l} \sum_{i=1}^{p_l} b_{l,i} c_\textup{BM}(1, \omega_{l, i}) + \sum_{i=1}^{p_l} \sum_{j=1}^{p_l} b_{l,i} b_{l,j} c_\textup{BM}(\omega_{l,i}, \omega_{l, j}) \right] \\
  &= 
  \frac{1}{\Delta_l^2} \left[ a_{l,p_l}^2 -2 a_{l,p_l} \sum_{i=1}^{p_l} b_{l,i} \omega_{l, i} + \sum_{i=1}^{p_l} \sum_{j=1}^{p_l} b_{l,i} b_{l,j} c_\textup{BM}(\omega_{l,i}, \omega_{l, j}) \right].
  \end{split}
\end{talign*}
We can compute the norm $\| f_l - f_{l-1} \|_{\mathcal{H}_l}^2$ in a similar way.
\begin{talign*}
  \| f_l  - f_{l-1}\|_{\mathcal{H}_l}^2 
  &= 
  \| f_l\|_{\mathcal{H}_l}^2  + \| f_{l-1}\|_{\mathcal{H}_l}^2 + \frac{2}{\Delta_l\Delta_{l-1}} \Big[- a_{l,p_l}a_{l-1,p_{l-1}} + a_{l,p_l} \sum_{i=1}^{p_{l-1}} b_{l-1,i} c_\textup{BM}(1, \omega_{l-1, i})\\
  & \qquad \qquad \qquad + a_{l-1,p_{l-1}} \sum_{i=1}^{p_l} b_{l,i} c_\textup{BM}(1, \omega_{l, i}) - \sum_{i=1}^{p_{l-1}} \sum_{j=1}^{p_l} b_{l-1,i} b_{l,j} c_\textup{BM}(\omega_{l-1,i}, \omega_{l, j}) \Big].
\end{talign*}
We used a fixed grid to pick the quadrature point. This will mean that A5 is satisfied.  A1--A4 will be trivially satisfied, and A4 can be checked according to the derivation above. Since we used Mat\'ern 0.5 kernel, the smoothness is the same as the Brownian motion kernel, we discard the influence of unknown constants and calculate the optimal sample size for MLBQ.

\paragraph{Experimental Settings} 
We used a Mat\'ern kernel with smoothness $v=0.5$ and all the formulae for the kernel mean and initial error are provided in \Cref{appendix_kmie}.  The RKHS norms are given by $(\|f_0\|_{\tau},\|f_1-f_0\|_{\tau},\|f_2-f_1\|_{\tau}) = (62.5, 22.5, 3.125) \times 10^{-3}$, and the variance by $(V_0,V_1,V_2)=(1.305, 0.088, 0.002) \times 10^{-3} $. The number of evaluations at each level for different budget constraints are shown in Table~\ref{PE:samplesize}. These levels correspond to a very coarse finite element mesh, a moderately fine finite element mesh and a fine finite element mesh, respectively. The lengthscales were optimised using L-BFGS. For the illustration example shown in Figure \ref{fig:Illustration}, we used the same approximation for $f_0$ and $f_1$ as in the Poisson equation example. However, we used a more accurate approximation of $f_2$ ($C_2=168 \times 10^{-3}$ seconds) and all approximations were multiplied by 7 to make the difference between different levels more significant. For the illustration example, the number of evaluations at level 0, 1 and 2 was 16, 11 and 3 respectively. For BQ, we used 4 evaluations of $f_2$. We still used a fixed grid to pick the quadrature point in the illustration example.

\begin{table}[htbp]
\caption{Number of evaluations at level $l$ given budget constraint $T$.} \label{PE:samplesize}
 \begin{center}
\begin{tabular}{c|c|ccc}
& T  & \textbf{$l=0$}  & \textbf{$l=1$} & \textbf{$l=2$}  \\ \hline
\multirow{3}{1cm}{$n_l^{\text{MLBQ}}$} & 0.376s  & 38 & 15 & 3 \\ \cline{2-5} 
                  & 0.751s & 77 & 30 & 5\\ \cline{2-5} 
                  & 1.503s & 153 & 60 & 10\\ \hline
\multirow{3}{1cm}{$n_l^{\text{MLMC}}$} & 0.376s  & 67 & 11 & 1\\ \cline{2-5} 
                  & 0.751s &  133 & 23 & 2\\ \cline{2-5} 
                  & 1.503s & 266 & 46  & 3\\ \hline
\end{tabular}
\end{center}
\end{table}

\begin{figure}[h]
    \centering
    \includegraphics[width=0.5\textwidth]{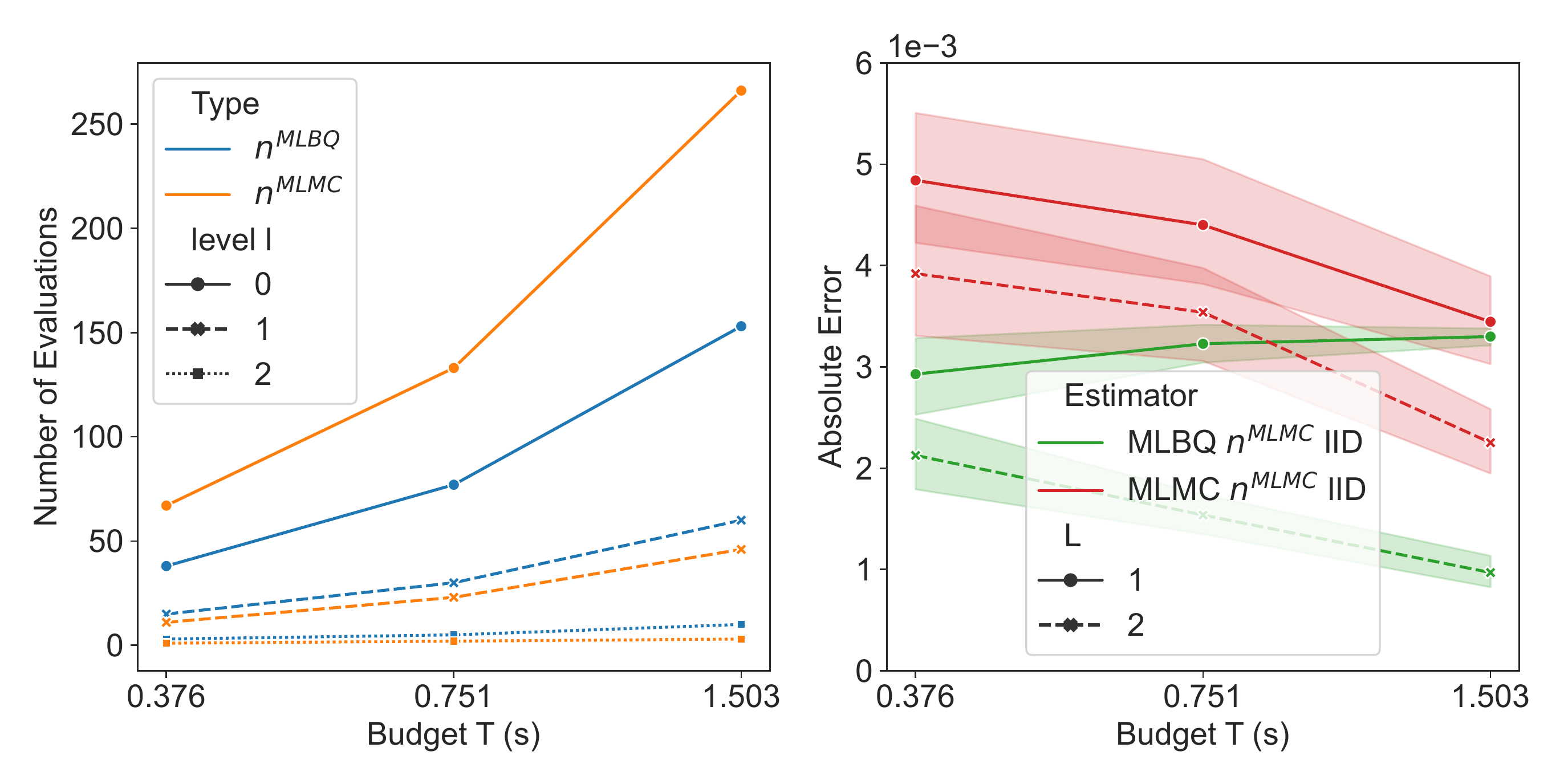}
    \caption{Poisson Equation. \textit{Left}: $n^{\text{MLMC}}$ and $n^{\text{MLBQ}}$ given budget constraint $T$. \textit{Right}:  Absolute integration error with different $L$. }
    \label{fig:PE3L_extra}
\end{figure}

\paragraph{Additional Results}  The number of evaluations at each level for different budget constraints are visualized in the left-hand side plot of Figure \ref{fig:PE3L_extra}. The right-hand plot shows the empirical mean and 95\% confidence interval of the absolute errors obtained with MLBQ and MLMC with IID points when we keep or remove the third level $\Pi[f_2-f_1]$. Benefiting from the fast convergence rate of BQ, the improvement of adding a few points in high-fidelity level (level 2) is significant for MLBQ but not for MLMC, especially when the budget constraint $T$ is small. 

\paragraph{Comparison to Multilevel Bayesian Quadrature with Separable Kernel}
We compare MLMC, MLBQ with a-priori independent $f_l-f_{l-1}$ (MLBQ),  and MLBQ with separable kernels (SK-MLBQ) (as in Appendix \ref{append:correlated_case}). We compare three different separable kernels, $B_1, B_2, B_3$, where $(B_1)_{i,i}=(B_2)_{i,i}=(B_3)_{i,i}=1$, and $(B_1)_{i,j}=0.01$, $(B_2)_{i,j}=0.05$ and $(B_3)_{i,j}=0.1$, for $i \neq j, i,j\in\{0,\ldots,L\}$. The computational cost of using SK-MLBQ depends on the budget constraint and the number of samples at each level. When $T=1.503$s, the computational cost of using SK-MLBQ with $n^{\text{MLMC}}$ is 0.374s, which is around 1.4 times that of using MLBQ with $n^{\text{MLMC}}$ (0.268s). The ratio will increase if we employ a larger budget constraint and use $n^{\text{MLBQ}}$. Figure \ref{fig:PE_SKBQ} visualizes the result of 100 repetitions of the experiment, where for each repetition, we evaluated $f_0,\ldots, f_L$ at new point sets, and used the same dataset for MLBQ, MLMC and SK-MLBQ to estimate $\Pi[f]$. 

The figure reveals that when a low cross-level correlation is established by using separable kernel $B_1$, the performance of SK-MLBQ improves slightly on that of MLBQ.  However, as the specified cross-level correlation increases, the performance of SK-MLBQ deteriorates and MLBQ outperforms SK-MLBQ when using $B_2$ and $B_3$.  Overall, this experiment shows that SK-MLBQ raises computational costs and does not ensure a significant reduction in error.

\begin{figure}[h]
    \centering
    \includegraphics[width=0.3\textwidth]{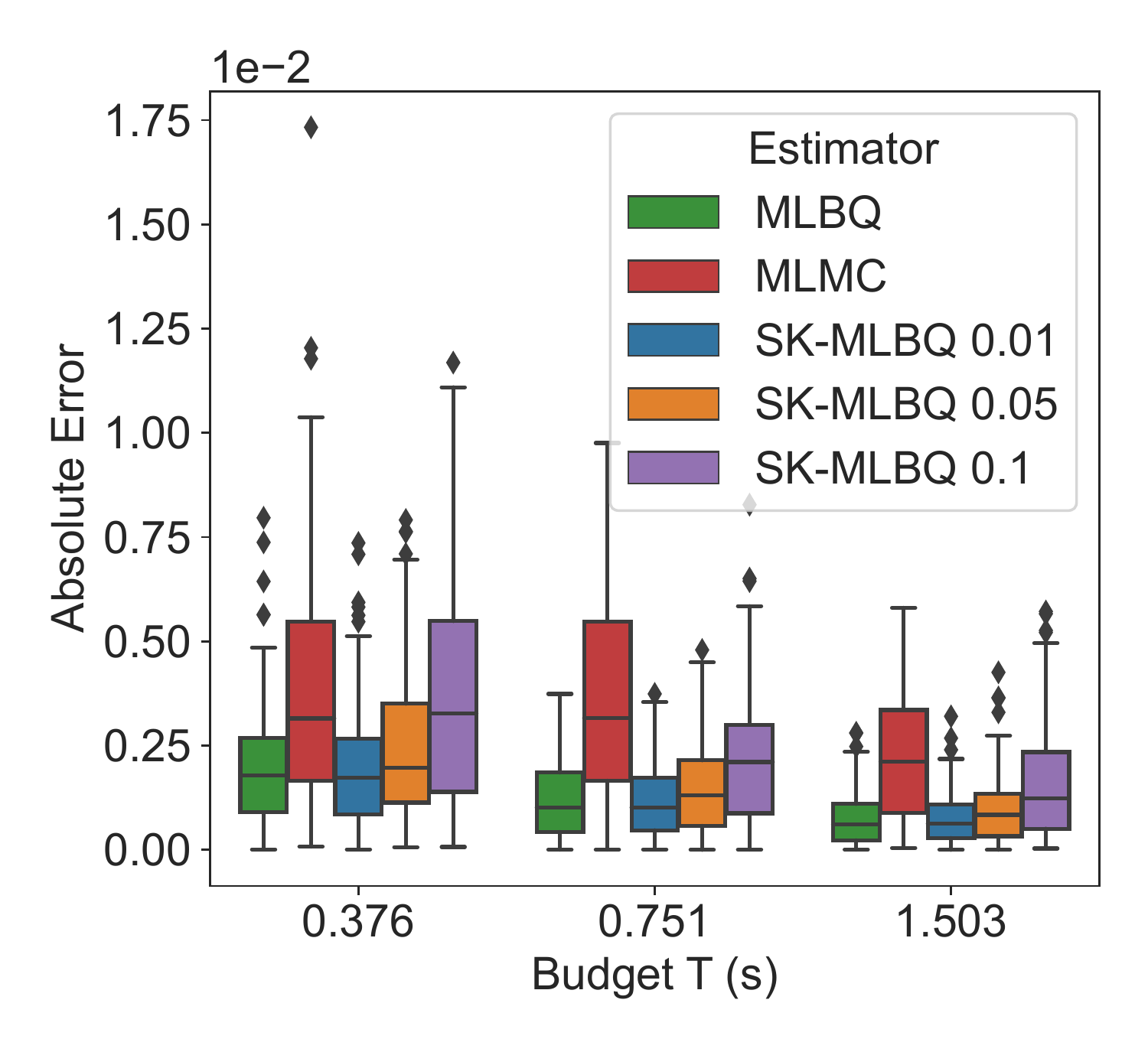}
    \caption{Poisson Equation.  Absolute integration error.  }
    \label{fig:PE_SKBQ}
\end{figure}

\subsection{Experiment 2: ODE with Random Coefficient and Forcing}\label{appendixODE_solver}

\paragraph{Construction of the Levels}  
We now provide details on the solver (finite difference approximation). We first expand the equation:
\begin{talign*}
  \frac{\mathrm{d}}{\mathrm{d} x}\left( c(x)\frac{\mathrm{d}u}{\mathrm{d} x} \right) &= -50^2 \omega_2^2\; 
   \Leftrightarrow \;
   \omega_1\frac{\mathrm{d}u}{\mathrm{d}x}+(1+\omega_1x)\frac{\mathrm{d}^2u}{dx^2}  =50\omega_2^2
\end{talign*}
 for $x\in (0,1)$. Let $u(x_i)=u(ih)$ for $i \in \{i,\ldots,(1-h)/h\}$ with $u(0)=u(1)=0$, we will approximate the left-hand side of the equation above using a finite difference approximation with spacing $h>0$:
\begin{talign*}
 \omega_1\frac{u(x_i)-u(x_{i}-h)}{h} + (1+\omega_1x_i)\frac{u(x_i+h)-2u(x_{i})+u(x_i-h)}{h^2} &  =50\omega_2^2\\ \omega_1\frac{u(x_i)-u(x_{i-1})}{h}+(1+\omega_1ih)\frac{u(x_{i+1})-2u(x_{i})+u(x_{i-1})}{h^2}&  =50\omega_2^2\\
\omega_1\frac{iu\left((i+1)h\right)-(2i-1)u(ih)+(i-1)u\left((i-1)h\right)}{h}+\frac{u\left((i+1)h\right)-2u(ih)+u\left((i-1)h\right)}{h^2}&  =50\omega_2^2
\end{talign*}
Then, bringing the random coefficient and the random forcing into consideration, the approximation at level $l$ is
\begin{talign*}
  f_l(\omega)=\sum_{i=1}^{1/h_l-1} h_lu(ih_l,\omega), 
\end{talign*}
where $u_l=(u(h_l,\omega),u(2h_l,\omega),\ldots,u(1-h_l,\omega))^\top$ can be solved from the linear system $(\omega_1 Q_l/h_l + L_l/h_l^2 ) u_l = 50\omega_2^2 \mathbf{1}$, where $\mathbf{1} \in \mathbb{R}^{(1-h_l)/h_l }$ is a vector of ones, $Q_l \in \mathbb{R}^{(1-h_l)/h_l \times (1-h_l)/h_l}$ is a tridiagonal stiffness matrix with
\begin{talign*}
  (Q_l)_{i,i}  = -2i+1, \qquad
  (Q_l)_{i,i-1}  = (Q_l)_{i-1,i} = i-1 ,
\end{talign*}
and $L_l \in \mathbb{R}^{(1-h_l)/h_l \times (1-h_l)/h_l}$ is a tridiagonal stiffness matrix with
\begin{talign*}
  (L_l)_{i,i}  = -2, \qquad (L_l)_{i,i-1}  = (L_l)_{i-1,i} = 1.
\end{talign*}

\paragraph{Experimental Setting}
Table \ref{ODE:samplesize} lists the number of evaluations at each level for multilevel estimators and BQ under different budget constraints. Three levels correspond to a very coarse ODE solver, a moderately fine ODE solver and a fine ODE solver, respectively. In this experiment, we used tensor product Mat\'ern kernel with smoothness $v=2.5$ and squared exponential kernel. The related analytical formulae are provided in 
\Cref{appendix_kmie}, and the Adam optimiser was used to select lengthscales. It is worth mentioning that the closed form of the initial error of the Mat\'ern kernel with respect to Gaussian distributed random variables doesn't exist. Since the closed form equation of the kernel mean in this case is known, we can estimate the initial error very precisely and efficiently with MC estimator by using a large number of IID samples from the Gaussian distribution.

\begin{table}[htbp]
\caption{Number of evaluations for multilevel estimators and BQ given budget constraint $T$.}\label{ODE:samplesize}
\begin{center}
\begin{tabular}{c|ccc|c}
$T$ & $l=0$  & $l=1$ & $l=2$  & \textbf{BQ} \\ \hline
0.303s & 166 & 27 & 3  & 15\\ \hline
1.517s & 830 & 135 & 15 & 75 \\ \hline
30.347s & 16579 & 2701 & 308 & /\\ \hline
151.736s & 82984 & 13505 & 1538 & /\\ \hline
\end{tabular}
 \end{center}
\end{table}

\paragraph{Additional Results} Figure \ref{fig:ODE3L_extra} shows the empirical mean and 95\% confidence interval of the absolute errors obtained with MLBQ and MLMC with IID points when we keep or remove the third level $\Pi[f_2-f_1]$. Benefiting from the fast convergence rate of BQ, the improvement of adding a few points in high-fidelity level (level 2) is significant in small budget cases (budget T = 0.303s) for MLBQ but not for MLMC.

\begin{figure}[h]
    \centering
    \includegraphics[width=0.25\textwidth]{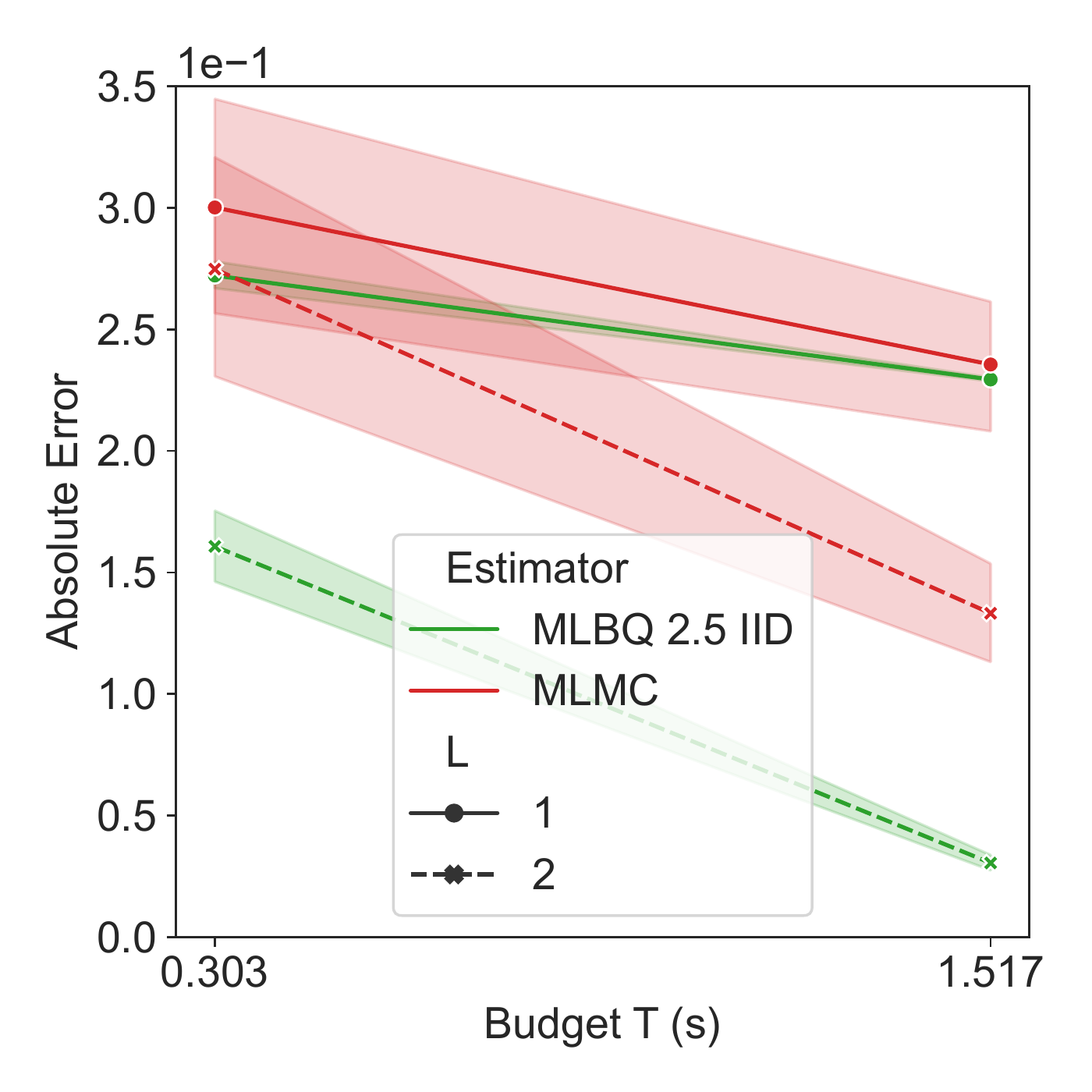}
    \caption{ODE with Random Coefficient and Forcing: Absolute integration error with different $L$.}
    \label{fig:ODE3L_extra}
\end{figure}

\subsection{Experiment 3: Landslide-Generated Tsunami}\label{appendixlandslide}

\paragraph{Construction of the Levels} 
To model the tsunami wave, Volna-OP2 \citep{giles2020performance} numerically solves the nonlinear shallow water equations:
\begin{talign*} \label{NSWE}
  \frac{\partial H}{\partial t}+\nabla \cdot (H \vec{v}) &=0, \qquad
    \frac{\partial H\vec{v}}{\partial t}+\nabla \cdot (H \vec{v}\otimes \vec{v} + \frac{g}{2}H^2I_2)=gH\nabla h,
\end{talign*}
where $\vec{v} = (u,v)$ is the depth averaged horizontal velocities, $g$ is the acceleration due to the gravity, $h$ is the underwater topography (bathymetry), $\eta$ is the wave height measured from the sea level at rest (free surface elevation) and therefore $H=h+\eta$ is the total water depth. $I_2$ denotes a $2\times2$ identity matrix. The exact form of the time dependent bathymetry of the landslide case was introduced in \cite{lynett2005numerical} but is reproduced here. The bathymetry of the sloping beach is prescribed by $h(x,t) = h_S(x,t)-x \tan(\omega_2)$, where $h_S$ is the profile of the sliding mass 
\begin{talign*}
    h_S(x,t)=\omega_1 d_o \frac{\left[1+\tanh{\left(\frac{2(x-x_l(t))}{\omega_3}\right)}\right]\left[1+\tanh{\left(\frac{2(x-x_r(t))}{\omega_3}\right)}\right]}{\left[1+\tanh{(\cos{(\omega_2)})}\right]\left[1-\tanh{(-\cos{(\omega_2)})}\right]}.
\end{talign*}
 The right and left boundaries of the slide are called $x_r(t)$ and $x_l(t)$ respectively and given by:
\begin{talign*}
x_r(t) =x_c(t)+\frac{\omega_3}{2}\cos(\omega_2), 
\qquad
x_l(t) =x_c(t)-\frac{\omega_3}{2}\cos(\omega_2),
\end{talign*}
where $x_c$ is the horizontal location of the center point of the slide. By assuming the initial depth of the center point of the slide to be $50$m, $x_c$ is given by: $x_c(t)= 50 /\tan(\omega_2)+u_s t$, where $u_s$ is the time-dependent velocity of the submerged landslide and can be calculated by
\begin{talign*}
u_s = u_t \tan\left(\frac{t}{t_0} \right),
\qquad
u_t = \sqrt{g \omega_3\frac{\pi}{2}\sin({\omega_2})},
\qquad
t_0 = u_t\frac{2}{g\sin({\omega_2})}.
\end{talign*}
Volna-OP2 uses a finite volume method with two dimensional meshes and thus the set up (Figure \ref{fig:sketch}, left) is extended in the $y$ direction, which is perpendicular to the page and results in translational symmetry along this axis. The domain of $y \in  [-\Delta x, \Delta x]$, where $\Delta x$ is the spatial resolution. A representative example of the solution as given by Volna-OP2 for $\vec{\omega}=(0.375,10^{\circ},150$ m$)$ is presented in Figure \ref{fig:snapshot}. The sub figures showcase transects of the bathymetry and surface elevation at various time points.
 
\begin{figure}[h]
    \centering
    {\includegraphics[width=0.8\linewidth]{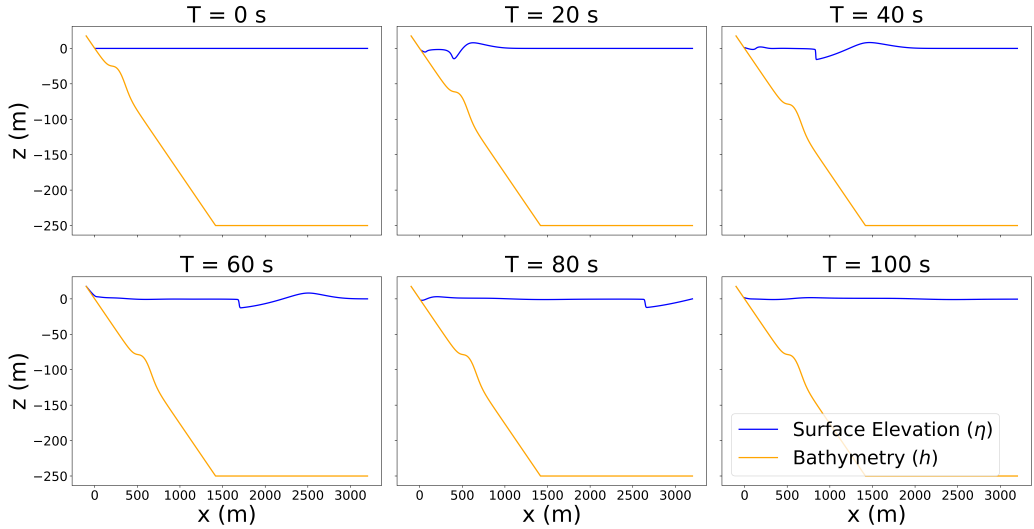}}%
    \caption{ Snapshots of the bathymetry and the solution of the PDE, surface evaluation $\eta(x,t,\omega)$ at various timestamps as given by Volna-OP2 when $\vec{\omega}=(0.375,10^{\circ},150 m)$.}
    \label{fig:snapshot}%
\end{figure}

\begin{table}[htbp]
\caption{Number of evaluations at level $l$ given budget constraint $T$.} \label{Tsu:samplesize}
\begin{center}
\begin{tabular}{c|ccccc}
$T$ & $l=0$  & $l=1$ & $l=2$ & $l=3$ & $l=4$ \\ \hline
1200s&32 &16 & 8 & 4 & 2 \\ \hline
 6000s&160 & 80 & 40 & 20  &10 \\ \hline
12000s &320 & 160 & 80 & 40  &20 \\ \hline
\end{tabular}
\end{center}
\end{table}

\paragraph{Experimental Design}
The number of evaluations at each level for different budget constraint is listed in Table \ref{Tsu:samplesize}. The lengthscales were optimised using the Adam optimiser.

\subsection{Experiment 4: Linear Function}\label{appendixstep}
In this section, we explore the impact of breaking our theoretical assumptions A1-A6 through a linear function example. Assumption~A1, A2, A6 could be generalised or replaced without affecting the convergence rate e.g. assuming bounded domain with Lipschitz boundary, and satisfying an interior cone condition \citep{wynne2021convergence} and assuming non-zero prior means \citep{teckentrup2020convergence}. Assumption A3-A5 are more crucial for obtaining the desired convergence rate. Assumption A3 is satisfied by using Mat\'ern kernels. Assumption A4 depends on the smoothness of $f_l$. If $f_l$ is infinitely differentiable, a convergence rate of MLBQ with squared exponential kernel can be derived by generalising Theorem~2.20 of \citet{karvonen2019kernel}.
Assumption A5 ensures that we use a dataset that covers the domain well, which is important to the convergence rate of MLBQ.

The integrand we consider is the following linear function,
\begin{talign*}
  f(\omega) &= \omega \: \text{ for } \:  \omega \in [0,10].
\end{talign*}
The integral of interest is $\Pi[f]=0.1\int^{10}_0 f(\omega)d\omega$, so $\Pi$ is a $\text{Unif}(0,10)$. We use step functions $f_0, f_1$ and $f_2$ as our approximations to $f$. Let $p_l \in \mathbb{N}$ and $0=\omega_{l,1}<\omega_{l,2}<\cdots<\omega_{l,p_l}=10$, then for $i=2,\ldots,p_l$
\begin{talign*}
 f_l(\omega)= 
   \begin{cases}
    \frac{\omega_{l,i-1}+\omega_{l,i}}{2} \: & \text{ if } \: \omega \in [\omega_{l,i-1},\omega_{l,i}) , \\
     \frac{\omega_{l,p_l-1}+10}{2} \: & \text{ if } \: \omega= 10.
  \end{cases}
\end{talign*}

We compare different settings to assess the impact of breaking our assumptions. Assumption~A1, A2, A4 and A6 are upheld for all settings in this example. We compare two different kernels, the Mat\'ern $0.5$ kernel and the squared exponential kernel, which violates Assumption~A3. Assumption~A5 is violated by utilizing IID points and a bad experimental design, in which 90\% of the points are sampled from $\text{Unif}(0,5)$, and the remaining 10\% are sampled from $\text{Unif}(5,10)$. We use $n^\textup{MLMC}$ for all settings in this example. Sample sizes are given in Table \ref{step:samplesize}.

\begin{table}[htbp]
\caption{Number of evaluations at level $l$ given budget constraint $T$.} \label{step:samplesize}
\begin{center}
\begin{tabular}{c|ccc}
$T$ & $l=0$  & $l=1$ & $l=2$ \\ \hline
0.002s&37 &8 & 2  \\ \hline
0.004s&74 & 15 & 4 \\ \hline
\end{tabular}
\end{center}
\end{table}

Figure \ref{fig:step} visualizes the approximations $f_0$, $f_1$ and $f_2$ and the results of 100 repetitions of the experiment. The right-hand side plot shows that MLBQ with Mat\'ern $0.5$ kernels and IID points significantly outperforms the others. When the budget is small, MLBQ with squared exponential kernels and IID points exhibits a slight advantage over MLMC. However, when the budget is increased, their performance becomes close. When the budget is small, MLMC demonstrate a slight superiority over MLBQ with Mat\'ern $0.5$ kernels and the bad experimental design, but when the budget is larger, MLBQ with Mat\'ern $0.5$ kernels and the bad experimental design performs slighly better than MLMC. MLMC always outperforms MLBQ with squared exponential kernels and the bad experimental design. This example highlights the substantial impact that the choice of kernels and experimental design have on the performance of MLBQ.

\begin{figure}
    \centering
    {\includegraphics[width=0.55\linewidth]{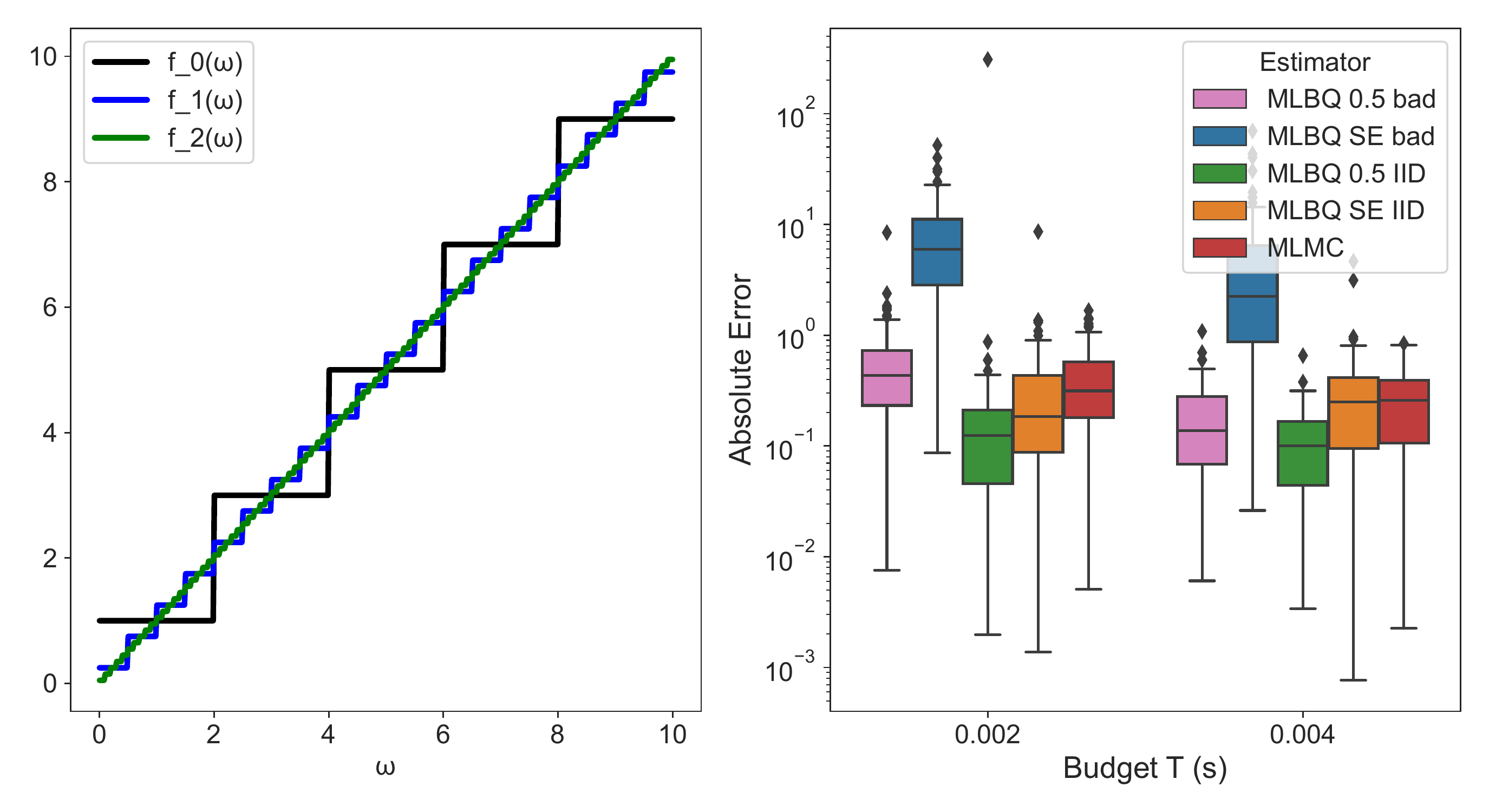}}%
    \caption{Linear Function. \textit{Left:} The approximations to $f$. \textit{Right:} Absolute integration error.}
    \label{fig:step}%
\end{figure}

 \subsection{Experiment 5: Model Evidence}\label{appendixME}

 A reviewer recommended approximating model evidence integrals based on likelihoods using datasets of different sizes. The reviewer mentioned that multilevel methods could help here, but unfortunately our initial attempts at using MLBQ in this context did not provide convincing results. This problem might therefore warrant some more efforts in future work.

\subsection{Analytical Formulae for the Kernel Means and Initial Errors}\label{appendix_kmie}

As mentioned in the main text, BQ algorithms usually require integrals of the kernel in closed-form. In this section, we provide all of the closed-form formulae used in the paper.

\paragraph{Mat\'ern Covariance Function with Smoothness $v = 1/2$}

When $\Pi$ is a uniform distribution on some interval $[a,b]$ and the covariance function is Mat\'ern covariance function with smoothness $v=0.5$ and length-scale $\gamma$, the kernel mean can be computed analytically as 
\begin{talign*}
     \Pi [c_{{1/2}}(\cdot, \omega)]&=\int^{b}_{a}  \exp\left(-\frac{|\omega-\omega^\prime|}{\gamma}\right) (b-a)^{-1} d\omega^{\prime} = (b-a)^{-1} \left( 2\gamma-\exp(\frac{a-\omega}{\gamma})\gamma-\exp(\frac{-b+\omega}{\gamma})\gamma\right),
\end{talign*}   
and the initial error can be computed analytically as
\begin{talign*}
\Pi \Pi [c_{{1/2}}(\cdot, \cdot)]&=\int^{b}_{a}\int^{b}_{a}\exp\left(-\frac{|\omega-\omega^\prime|}{\gamma}\right)(b-a)^{-1} d\omega^{\prime}(b-a)^{-1}d\omega
=2\gamma \left (b-a-\gamma+\exp\left(\frac{a-b}{\gamma}\right)\gamma\right) (b-a)^{-2}.
\end{talign*}

\paragraph{Mat\'ern Covariance Function with Smoothness $v=5/2$}

When $\Pi$ is a standard Gaussian distribution and the covariance function is Mat\'ern covariance function with smoothness $v=2.5$ and length-scale $\gamma$, the kernel mean can be computed analytically as 
\begin{talign*}
     \Pi[c_{{5/2}}(\cdot, \omega)]
     &=\int^{+\infty}_{-\infty}
     \left(1+\frac{\sqrt{5}\|\omega-\omega^\prime\|_2}{\gamma}+\frac{5\|\omega-\omega^\prime\|_2^2}{3\gamma^2}\right)
     \exp\left(-\frac{\sqrt{5}\|\omega-\omega^\prime\|_2}{\gamma}\right)
     \frac{1}{\sqrt{2\pi}}\exp\left(-\frac{1}{2}\omega^{\prime}{}^2\right)d\omega^{\prime}\\
    &=\exp \left(-\frac{\omega^2}{2}\right)
    \Big(4\sqrt{5}\gamma(-5+3\gamma^2)+\sqrt{2\pi}\Big(\exp(\frac{(\sqrt{5}-\gamma\omega)^2}{2\gamma^2})(25+3\gamma^4-10\sqrt{5}\gamma\omega\\
     &+3\sqrt{5}\gamma^3\omega+5\gamma^2(-2+\omega^2))\text{erfc}\left(\frac{\frac{\sqrt{5}}{\gamma}-\omega}{\sqrt{2}}\right)+\\
    &\exp\left(\frac{(\sqrt{5}+\gamma\omega)^2}{2\gamma^2}\right) \left(25+3\gamma^4+10\sqrt{5}\gamma\omega-3\sqrt{5}\gamma^3\omega+5\gamma^2(-2+\omega^2)\right)\text{erfc}\left(\frac{\frac{\sqrt{5}}{\gamma}+\omega}{\sqrt{2}}\right)\Big)\Big)/(6\gamma^4\sqrt{2\pi})\\
\end{talign*}
but the initial error cannot be computed analytically. 

When $\Pi$ is a uniform distribution on some interval $[a,b]$ and the covariance function is Mat\'ern covariance function with smoothness $v=2.5$ and length-scale $\gamma$, the kernel mean can be computed analytically as 
\begin{talign*}
         \Pi[c_{{5/2}}(\cdot, \omega)]
         &=\int^{b}_{a}
         \left(1+\frac{\sqrt{5}\|\omega-\omega^\prime\|_2}{\gamma}+\frac{5\|\omega-\omega^\prime\|_2^2}{3\gamma^2}\right)
         \exp\left(-\frac{\sqrt{5}\|\omega-\omega^\prime\|_2}{\gamma} \right)(b-a)^{-1}d\omega^{\prime}\\
    &=(b-a)^{-1} \Big(16\sqrt{5}\gamma-\exp\left(\frac{\sqrt{5}(a-\omega)}{\gamma}\right)\left(\frac{\sqrt{5}(8\gamma^2+5(a-\omega)^2)}{\gamma}+25(\omega-a)\right)\\
    & \qquad \qquad
+\exp\left(\frac{\sqrt{5}(-b+\omega)}{\gamma}\right)\left(-\frac{\sqrt{5}(8\gamma^2+5(b-\omega)^2)}{\gamma}+25(-b+\omega)\right) \Big)/15 .
\end{talign*}    
and the initial error can be computed analytically as
\begin{talign*}
   \Pi\Pi[c_{{5/2}}(\cdot, \cdot)]&=\int^{b}_{a}\int^{b}_{a} 
   \left(1+\frac{\sqrt{5}\|\omega-\omega^\prime\|_2}{\gamma}+\frac{5\|\omega-\omega^\prime\|_2^2}{3\gamma^2}\right)
   \exp\left(-\frac{\sqrt{5}\|\omega-\omega^\prime\|_2}{\gamma}\right)
   (b-a)^{-1}d\omega^{\prime}(b-a)^{-1}d\omega\\
    &=2 (b-a)^{-2} \left(8\sqrt{5}(b-a)\gamma-15\gamma^2+\exp\left(-\frac{\sqrt{5}(b-a)}{\gamma}\right)\left(5(b-a)^2+7\sqrt{5}(b-a)\gamma+15\gamma^2\right)\right)/15.
\end{talign*}

\paragraph{Squared Exponential Covariance Function}
When $\Pi$ is a uniform distribution on some interval $[a,b]$ and the covariance function is squared exponential with length-scale $\gamma$, the kernel mean can be computed analytically as 
  \begin{talign*}
    \Pi [c_{\textup{SE}}(\cdot, \omega)]
    &=\int^{b}_a \exp\left(-\frac{(\omega-\omega^\prime)^2}{\gamma^2}\right)(b-a)^{-1}d\omega^{\prime}\\
    &=(b-a)^{-1}\sqrt{\pi}\gamma(\text{erf}(\frac{\omega-a}{\gamma})+\text{erf}(\frac{b-\omega}{\gamma}))/2,
\end{talign*}
and the initial error can be computed analytically as
\begin{talign*}
      \Pi \Pi[c_{\textup{SE}}(\cdot, \cdot)]
      &=\int^{b}_a\int^{b}_a \exp\left(-\frac{(\omega-\omega^\prime)^2}{\gamma^2}\right)(b-a)^{-1}d\omega^{\prime}(b-a)^{-1}d\omega\\
      &=\gamma\left(\left(-1+\exp\left(-\frac{(a-b)^2}{\gamma^2}\right)\right)\gamma+(a-b)\sqrt{\pi}\text{erf}(\frac{a-b}{\gamma})\right)(b-a)^{-2}.
\end{talign*}
When $\Pi$ is a standard Gaussian distribution and the covariance function is squared exponential covariance function with length-scale $\gamma$, the kernel mean can be computed analytically as 
 \begin{talign*}
    \Pi [c_{\textup{SE}}(\cdot, \omega)]
    &=\int^{+\infty}_{-\infty}\exp(-\frac{(\omega-\omega^\prime)^2}{\gamma^2})\frac{1}{\sqrt{2\pi}}\exp(-\frac{1}{2}\omega^{\prime}{}^2)d\omega^{\prime}\\
    &=\gamma \exp(-\frac{\omega^2}{\gamma^2+2})(\gamma^2+2)^{-1/2},
\end{talign*}
and the initial error can be computed analytically as
\begin{talign*}
      \Pi \Pi[c_{\textup{SE}}(\cdot, \cdot)]
      &=\int^{+\infty}_{-\infty}\int^{+\infty}_{-\infty}\exp\left(-\frac{(\omega-\omega^\prime)^2}{\gamma^2}\right)\frac{1}{\sqrt{2\pi}}\exp(-\frac{1}{2}\omega^{\prime}{}^2)d\omega^{\prime}\frac{1}{\sqrt{2\pi}}\exp(-\frac{1}{2}\omega^2)d\omega  \\
      &=\gamma(\gamma^2+4)^{-1/2} .
\end{talign*}

\vfill

\end{document}